\documentclass{article}
\usepackage[]{arxiv}

\usepackage[T1]{fontenc}
\usepackage[utf8]{inputenc}

\usepackage{graphicx}
\usepackage{amsmath,amssymb,mathtools}
\usepackage{booktabs}
 \usepackage{amsthm}
 \usepackage{algorithm,algorithmic}
\usepackage[numbers]{natbib}
\usepackage{glossaries}
\usepackage{tikz}
\usetikzlibrary{shapes,arrows,positioning,decorations.pathreplacing,calc,backgrounds}

\newtheorem{assumption}{Assumption}
\newtheorem{definition}{Definition}
\newtheorem{lemma}{Lemma}
\newtheorem{theorem}{Theorem}
\newtheorem{proposition}{Proposition}
\newtheorem{corollary}{Corollary}
\newtheorem{remark}{Remark}
\newacronym{ctmc}{CTMC}{Continuous Time Markov Chain}
\newacronym{hji}{HJI}{Hamilton-Jacobi-Isaacs}
\newacronym{pdmp}{PDMP}{Piecewise Deterministic Markov Process}

\usepackage{xcolor}
\long\def\rev#1{#1}

\title{A Games-in-Games Paradigm for Strategic Hybrid Jump-Diffusions: Hamilton-Jacobi-Isaacs Hierarchy and Spectral Structure}

\author{Yunian Pan \& Quanyan Zhu\\
Department of Electrical \& Computer Engineering\\
New York University\\
370 Jay St, Brooklyn, NY 11201, USA}

\begin{document}

\maketitle

\begin{abstract}
\rev{This paper develops a hierarchical games-in-games control architecture for hybrid stochastic systems governed by regime-switching jump-diffusions. We model the interplay between continuous state dynamics and discrete mode transitions as a bilevel differential game: an inner layer solves a robust stochastic control problem within each regime, while a strategic outer layer modulates the transition intensities of the underlying Markov chain. A Dynkin-based analysis yields a system of coupled Hamilton-Jacobi-Isaacs (HJI) equations, for which we provide a feedback Stackelberg verification principle and establish well-posedness of the resulting coupled Riccati flow. In the structured Linear-Quadratic and Exponential-Affine specializations the hierarchy reduces to tractable coupled differential systems; this reduction recovers the classical Lyapunov--Metzler and Riccati--Metzler frameworks as special cases, and we characterize the spectral dissipation of the closed-loop generator selected by the outer game. The framework is demonstrated through a case study on adversarial market microstructure, where a risk-aware quoting rule adjusts inventory spreads against latent regime and predatory risks.}
\end{abstract}


\section{Introduction}

Hybrid systems, characterized by the interplay between continuous state dynamics and discrete event transitions, constitute a fundamental modeling paradigm for complex socio-economic engineering systems. Within this broad class, \textit{regime-switching jump-diffusions} occupy a central role, capturing systems where continuous stochastic trajectories are modulated by a hidden or observable Markov chain. Applications range from fault-tolerant control in networked systems \cite{YinZhu2010,pan2023resilience,pan2022poisoned} and cyber-physical security \cite{ZhuBasar2015,Pasqualetti2015Geometric,pan2023stochastic} to economic systems \cite{MaoYuan2006}. \rev{Beyond control within a fixed jump structure, a substantial body of work secures estimation and control for uncertain, switched, and networked systems, including asynchronous $H_\infty$ filtering for Markov jump systems~\citep{FangRenWangStojanovicHe2024} and robust identification of stochastic nonlinear parameter-varying models~\citep{MoratoStojanovic2021}, together with resilient consensus tracking and vibration suppression for multi-agent helicopter slung-load systems~\citep{PengSongSongTejadoStojanovic2026}. These approaches secure estimation and tracking against noise, faults, and adversarial perturbations under a given switching law; the present paper is complementary, placing the \emph{regime-transition law itself} inside a two-player zero-sum game, so that robustness is sought against an adversary shaping the discrete mode dynamics in addition to the continuous trajectory.}

In these settings, decision-making is rarely monolithic. It operates hierarchically: fast-timescale controllers regulate the continuous state (e.g., stabilizing a plant or hedging a portfolio), while slow-timescale policies influence the discrete operating modes (e.g., system reconfiguration or regime induction). However, classical literature typically decouples these layers. The theory of \textit{Piecewise-Deterministic Markov Processes} (PDMPs) \cite{Davis1984PDMP,CostaStability1999} and switching diffusions \cite{YinZhu2010} generally treats the regime transition mechanism as either exogenous (governed by nature) or subject to a single controller's optimization (optimal switching) \cite{Kharroubi2010Optimal}.

A critical gap exists in modeling \textit{adversarial hybrid interactions}, where the regime transitions themselves are the outcome of a strategic game. For instance, in cyber-physical systems, an attacker may seek to destabilize the system by inducing transitions to vulnerable modes \cite{ChenVaidya2018}, while a defender attempts to harden the transition logic. Existing differential game theory \cite{BasarOlsder1999,Buckdahn2008BSDE} provides robust tools for the continuous layer but typically assumes a fixed or purely stochastic discrete structure. Conversely, impulse games \cite{Basei2022Nonzero} focus on discrete interventions but often abstract away the continuous-time feedback loops.

This paper bridges this gap by developing a hierarchical \textit{Games-in-Games} control architecture for regime-switching jump-diffusions. We construct a bilevel system where a fast inner layer solves a robust stochastic differential game within each mode, while a strategic outer layer actively modulates the \textit{transition intensity kernel} of the underlying Markov chain. This structure, illustrated in Figure~\ref{fig:gig-core}, formalizes the problem of \textit{strategic regime control}, applicable to scenarios ranging from adversarial market microstructure to multi-modal resilient control.

\begin{figure}[htbp]
    \centering
    \includegraphics[width=.85\linewidth]{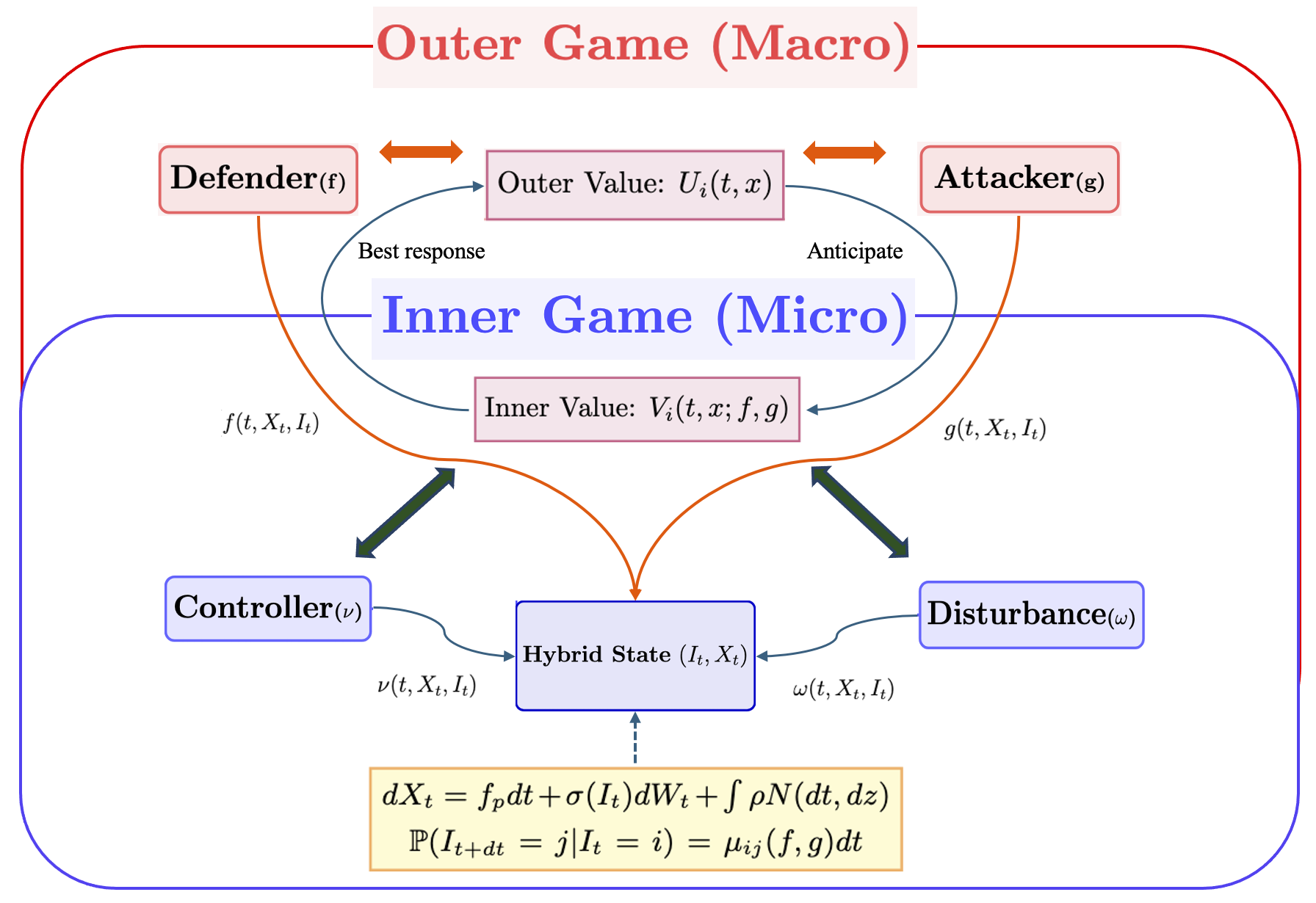}
    \caption{Compact Games-in-Games diagram: macro players shape regime switching; micro controls act on continuous dynamics under the active regime.}
    \label{fig:gig-core}
\end{figure}

Our contributions are threefold:

\begin{enumerate}
    \item  We provide a unified formulation for the games-in-games architecture on jump-diffusion spaces. By leveraging a unified Dynkin formula for switching diffusions \cite{YinZhu2010}, we decompose the bilevel problem into a hierarchy of coupled Hamilton-Jacobi-Isaacs (HJI) equations and provide a feedback Stackelberg verification principle under the stated Isaacs and regularity conditions. This separates the continuous control (inner Isaacs equation) from the strategic switching (outer Hamilton-Jacobi equation) while keeping the cross-layer dependence explicit.
    
    \item  While general HJI equations are computationally demanding, we prove that for the class of Linear-Quadratic (LQ) and Exponential-Affine transformations (as in the market microstructure case study), the hierarchy collapses into a system of \textit{coupled matrix differential equations}, for which we establish symmetry, well-posedness, and a sufficient condition for global existence on a finite horizon. This extends the classical coupled Riccati theory \cite{Dragan2001Riccati,CostaMarkovJump2006} to the game-theoretic setting with endogenous transition rates; the coupled flow recovers the Lyapunov--Metzler and Riccati--Metzler frameworks as its non-adversarial, frozen-rate special cases, and we characterize the spectral dissipation of the closed-loop generator selected by the outer game.
    
    \item  We demonstrate the framework's efficacy through an inventory game case study on adversarial market making. We derive a Risk Isomorphism principle that allows us to characterize equilibrium policies that pre-emptively adjust control gains (inventory spreads) based on the stability gap between regimes.
\end{enumerate}

The remainder of this paper is organized as follows. Section \ref{sec:problem} formulates the two-layer hybrid game. Section \ref{sec:solution} derives the viscosity solution hierarchy and the Dynkin transformation. Section \ref{sec:case} develops the coupled-Riccati solution and spectral structure for the LQ case. Section \ref{sec:application} details the market microstructure application, and Section \ref{sec:conclusion} concludes.

\section{Problem Formulation}\label{sec:problem}

We consider a hybrid decision architecture in which a continuous state evolves
according to mode-dependent stochastic dynamics, while a discrete mode process switches
between a finite collection of regimes. Two layers of strategic decision-making
interact: a fast-timescale controller/disturbance pair regulating the
continuous state, and a slow-timescale pair of agents whose actions influence
the transition rates among the discrete modes. This structure captures a wide
range of multi-layer hybrid systems, including resilient infrastructure networks,
multi-agent cyber-physical systems, and robust control under regime uncertainty.

Let $(\Omega, \mathcal{F}, \mathbb{P})$ be a complete probability space equipped with a filtration $\mathbb{F} = (\mathcal{F}_t)_{t \ge 0}$ satisfying the usual conditions of right-continuity and completeness. The time horizon is finite, $T < \infty$. 
We define the following primitive stochastic processes adapted to $\mathbb{F}$:
\begin{enumerate}
    \item $W = (W_t)_{t \ge 0}$ is a standard $d$-dimensional Brownian motion.
    \item \rev{$N(dt, dz)$ is an integer-valued (marked point) random measure on $[0, T] \times \mathcal{Z}$, where $\mathcal{Z} \subseteq \mathbb{R}^k$ is the space of jump marks (sizes), and $\tilde{N}(dt, dz)$ denotes its compensated counterpart. The jump intensity is \emph{control-dependent}; its precise action-dependent compensator, the resulting martingale measure $\tilde N$, and the action-free baseline intensity $\lambda(dz)$ are specified in Assumption~\ref{assum:regularity}(ii), once the control actions have been introduced in Definition~\ref{def:hybrid}.}
\end{enumerate}

\begin{definition}[Two-layer hybrid decision system]\label{def:hybrid}
Let $\mathcal{I}=\{1,\dots,N\}$ be a finite set of regimes (modes), $\mathbb{U}$ and $\mathbb{W}$ compact convex sets representing continuous-layer control and disturbance actions, and
$\mathcal{A}_D$, $\mathcal{A}_A$ finite sets of actions available to the two
agents governing the mode transitions.
A \emph{two-layer hybrid system} is a tuple
   $ \Gamma =\bigl(\mathcal{X},\,\mathcal{I},\,\mathcal{V},\,\mathcal{W},\,\mathcal{F},\,\mathcal{G}\bigr)$,
where:
\begin{enumerate}
\item $\mathcal{X} \subseteq \mathbb{R}^n$ is the continuous state space and $\mathcal{I}$ is the finite mode space;

\item The continuous-layer policies
\begin{equation*}
\nu:[0,T]\times \mathcal{X}\times \mathcal{I}\to \mathbb{U},\qquad
\omega:[0,T]\times \mathcal{X}\times \mathcal{I}\to \mathbb{W}
\end{equation*}
are Borel measurable and $\mathbb{F}$-progressively measurable; the associated admissible policy classes are denoted $\mathcal{V}$ and $\mathcal{W}$. \rev{Throughout, $\nu$ and $\omega$ denote the inner-layer feedback \emph{policies} (maps into the action sets), while $u:=\nu(t,X_t,I_t)\in\mathbb{U}$ and $w:=\omega(t,X_t,I_t)\in\mathbb{W}$ denote the corresponding realized \emph{action values}; the symbols are used consistently in this sense.}

\item The mode-selection policies
\begin{equation*}
f:[0,T]\times \mathcal{X}\times \mathcal{I}\to \Delta(\mathcal{A}_D),\qquad
g:[0,T]\times \mathcal{X}\times \mathcal{I}\to \Delta(\mathcal{A}_A)
\end{equation*}
assign mixed actions from $\mathcal{A}_D$ and $\mathcal{A}_A$. The admissible classes are denoted $\mathcal{F}$ and $\mathcal{G}$.
\end{enumerate}

Given $(f,g)$ and $(\nu,\omega)$, the hybrid state process $(X_t, I_t)_{t \in [0,T]}$ evolves as a \emph{Regime-Switching Jump-Diffusion} defined by:

\begin{itemize}
\item[(a)] Continuous dynamics: Between jumps of the mode $I_t$, the continuous state $X_t$ evolves according to the stochastic differential equation (SDE):
\begin{align}
dX_t &= f_p\bigl(t, X_t, \nu_t, \omega_t; I_t\bigr)dt
+ \sigma\bigl(t, X_t; I_t\bigr)dW_t \nonumber \\
&\quad + \int_{\mathcal{Z}} \rho (t, X_t, z; I_t) \tilde{N}(dt, dz),
\label{eq:cont-flow}
\end{align}
with $X_0 = x_0 \in \mathcal{X}$. Here, $\nu_t = \nu(t, X_t, I_t)$ and $\omega_t = \omega(t, X_t, I_t)$.

\item[(b)] Discrete dynamics: The mode process $I_t \in \mathcal{I}$ is a controlled continuous-time Markov chain with generator matrix $\Pi_t = [\mu_{ij}(t)]_{i,j \in \mathcal{I}}$ modulated by the outer policies:
\begin{equation}
\mathbb{P}\bigl(I_{t+dt} = j \mid I_t = i, X_t = x\bigr) = \mu_{ij}\bigl(f(t,x,i), g(t,x,i)\bigr)dt + o(dt),
\label{eq:mode-kernel}
\end{equation}
where $\mu_{ii} = - \sum_{j \neq i} \mu_{ij}$. The mode-selection policies $(f,g)$ determine the
regime transition intensities $\mu_{ij}: \Delta(\mathcal{A}_D) \times \Delta(\mathcal{A}_A) \to [0, \infty),$
yielding the instantaneous rate $\mu_{ij}\bigl(f(t,X_t,I_t),g(t,X_t,I_t)\bigr)$
at each $t$.
\end{itemize}

Above defines a measure $\mathbb{P}^{\nu,\omega,f,g}$ on the path space $D([0,T]; \mathcal{X} \times \mathcal{I})$.
\end{definition}
The performance of a policy tuple is evaluated through the cost functional:
\begin{equation}
J(f,g;\nu,\omega)
=
\mathbb{E}\!\left[
c_T\bigl(X_T, I_T\bigr)
+
\int_{0}^{T}
c\bigl(t, X_t, \nu_t, \omega_t, I_t\bigr)\,dt
\right].
\label{eq:cost}
\end{equation}

The two-layer decision architecture is expressed as the bi-level optimization
problem
\begin{equation}
    \label{eq:bilevel-problem}
    \rev{\begin{aligned}
         & \min _{f \in \mathcal{F}} \max _{g \in \mathcal{G}} \Phi\left(f, g ; \nu^{*}, \omega^{*}\right) \\
         & \text { s.t. } \quad
         J(f,g;\nu^*,\omega)
         \leq J(f,g;\nu^*,\omega^*)
         \leq J(f,g;\nu,\omega^*) \\
         & \hspace{18mm}\text{for all }(\nu,\omega)\in\mathcal{V}\times\mathcal{W}.
    \end{aligned}}
\end{equation}

where the inner minimax determines the continuous-layer value and the
outer minimax governs mode manipulation. The functional $\Phi\left(\cdot, \cdot ; \nu^{*}, \omega^{*}\right):  \Delta\left(\mathcal{A}_{D}\right) \times \Delta\left(\mathcal{A}_{A}\right) \rightarrow \mathbb{R}$ represents the outer-layer preference and will be specified
later.

\begin{assumption}[Standing assumptions]\label{assum:regularity}
\begin{itemize}
~
\item[(i)] (\emph{Generalized Isaacs Condition})
For each fixed mode $i \in \mathcal{I}$ and any smooth test function $\phi \in C^2(\mathcal{X})$, we define the generalized Hamiltonian $\mathcal{H}_i[\phi]$ acting on the state $x$, gradient $p = \nabla \phi(x)$, Hessian $M = \nabla^2 \phi(x)$, and controls $(u,w)$:
\begin{align*}
\mathcal{H}_i[\phi](t,x,p,M,u,w) &= c(t,x,u,w,i) + p^\top f_p(t,x,u,w; i) \\
&+ \frac{1}{2}\mathrm{Tr}\bigl(\sigma(t,x; i)\sigma(t,x; i)^\top M\bigr) \\
+ \int_{\mathcal{Z}} \Bigl( \phi(x+\rho (t,x,z; i)) & - \phi(x) - p^\top\rho (t,x,z; i) \Bigr) \eta(dz; u,w).
\end{align*}
We assume that the minimax condition holds for all valid inputs:
\begin{equation*}
\min_{u\in\mathbb{U}}\max_{w\in\mathbb{W}}\mathcal{H}_i[\phi](t,x,p,M,u,w) = \max_{w\in\mathbb{W}}\min_{u\in\mathbb{U}}\mathcal{H}_i[\phi](t,x,p,M,u,w).
\end{equation*}

\item[(ii)] The coefficients $f_p, \sigma, \rho $ satisfy the standard Lipschitz and linear growth conditions in $x$, uniformly in $(u,w)$. \rev{Here $\eta(dz;u,w)$ is the \emph{controlled jump-intensity kernel} of the random measure $N$ introduced in Section~\ref{sec:problem}: a finite measure on $\mathcal{Z}$ depending Borel-measurably on the realized actions $(u,w)$, which serves as the predictable compensator $\eta(dz;u_t,w_t)\,dt$ of $N$ under the induced law $\mathbb{P}^{\nu,\omega,f,g}$. The compensated (martingale) measure appearing in \eqref{eq:cont-flow} is accordingly
\[
\tilde N(dt,dz)=N(dt,dz)-\eta(dz;u_t,w_t)\,dt,
\]
and the action-free baseline intensity is $\lambda(dz):=\eta(dz;u^\circ,w^\circ)$; a single controlled kernel $\eta$ thus governs the jump term throughout.} The kernel $\eta(dz; u,w)$ is bounded and satisfies $\int_{\mathcal{Z}}(1\wedge\|z\|^2)\,\eta(dz;u,w)<\infty$ uniformly in $(u,w)$. The costs $c, c_T$ satisfy quadratic growth conditions.

\end{itemize}
\end{assumption}

\begin{remark}[Scope of the Isaacs conditions]\label{rem:isaacs-scope}
\rev{Assumption~\ref{assum:regularity}(i), and its outer-layer counterpart Assumption~\ref{ass:outer-isaacs}, are structural minimax-interchange requirements. They do not hold for arbitrary data, and they are restrictive in the presence of jumps. A standard sufficient condition for the inner layer is that the Hamiltonian $\mathcal{H}_i[\phi]$ be convex in the control $u$ and concave in the disturbance $w$ over the compact convex action sets $\mathbb{U},\mathbb{W}$; the order of $\min_u$ and $\max_w$ may then be interchanged by Sion's minimax theorem. Analogously, the outer-layer condition follows when $\mathcal{H}_{\mathrm{out}}$ is convex in the minimizing mixed action and concave in the maximizing mixed action on the finite simplices. The inner condition holds for the linear--quadratic class of Section~\ref{sec:case} (with $R_i\succ0$ and $S_i\succ0$), while the entropy-enriched outer payoff used in that specialization is strictly convex--concave. The scalar inner optimizations displayed for the exponential--affine application in Section~\ref{sec:application} likewise admit unique optimizers. Accordingly, the general hierarchy of Sections~\ref{sec:problem}--\ref{sec:solution} should be read as a template; outside these structured subclasses, each Isaacs condition must be verified case by case.}
\end{remark}

\begin{lemma}[Existence and Estimates]\label{lem:existence}
Under Assumption~\ref{assum:regularity}, for any admissible policies, the hybrid system \eqref{eq:cont-flow}--\eqref{eq:mode-kernel} admits a unique strong solution $(X_t, I_t)_{t \in [0,T]}$. Furthermore, for any $p \ge 1$, there exists a constant $C_p > 0$ such that:
\begin{equation*}
\mathbb{E}\left[\sup_{t\in[0,T]} \|X_t\|^p\right] \le C_p\bigl(1+\|x_0\|^p\bigr).
\end{equation*}
\end{lemma}

\begin{proof}
Since the transition rates are bounded, the mode process $I_t$ undergoes finitely many jumps in $[0, T]$ almost surely. Between any two jump times, the system evolves as a standard SDE with jumps. Under the Lipschitz and linear growth conditions (Assumption~\ref{assum:regularity}), a unique strong solution exists for each interval (see e.g., \citep{applebaum2009levy, oksendal2005stochastic}). We construct the global solution $(X_t)_{t\ge0}$ by concatenating these trajectory segments at the jump times of $I_t$.

We apply It\^o's formula to the function $\phi(x) = \|x\|^p$. The linear growth assumption on the drift $f_p$ and jump intensity implies that the generator is bounded by $\mathcal{L}\phi(x) \le C(1 + \|x\|^p)$.
To bound the expectation of the supremum, we handle the martingale terms (diffusion and compensated jumps) using the {Burkholder-Davis-Gundy (BDG) inequalities}, which control the maximum of the stochastic integrals.
Combining these bounds yields an integral inequality for $g(t) = \mathbb{E}[\sup_{s \le t} \|X_s\|^p]$. The final result follows immediately from {Gr\"onwall's inequality}.
\end{proof}

\section{Cross-Layer Viscosity Solution}\label{sec:solution}

Once the hybrid architecture $\Gamma$ in Definition~\ref{def:hybrid} is equipped
with Assumption~\ref{assum:regularity}, the coupled state process
$(X_t, I_t)$ is a Regime-Switching Jump-Diffusion. Between jumps of the discrete mode $I_t$, the
continuous state $X_t$ follows the stochastic evolution generated by the drift $f_p$, diffusion $\sigma$,
and jump measure defined in \eqref{eq:cont-flow}. 
For any probe function
\begin{equation*}
\phi:[0,T]\times \mathcal{I}\times\mathcal{X}\to\mathbb{R},\qquad
\phi(\cdot, i, \cdot) \in C^{1,2}([0,T] \times \mathcal{X}),
\end{equation*}
the infinitesimal generator of $(X_t, I_t)$ under continuous-layer policies
$(\nu,\omega)$ and mode-selection policies $(f,g)$ is given by the sum of the
diffusion generator, the inner jump generator, and the regime-switching operator:
\begin{equation}
\label{eq:generator-full}
\begin{aligned}
(\mathcal{L}^{f,g,\nu,\omega}\phi)(t,x,i)
&=
\frac{\partial \phi}{\partial t}(t,x,i)
+ \nabla_x\phi(t,x,i)^{\top} f_p\bigl(t,x,\nu(t,x,i),\omega(t,x,i);i\bigr) \\
&\quad
+ \frac{1}{2}\mathrm{Tr}\Bigl(\sigma(t,x;i)\sigma(t,x;i)^\top \nabla^2_{xx}\phi(t,x,i)\Bigr) \\
+ \int_{\mathcal{Z}} \Bigl( &\phi(t, x+\rho (t,x,z; i), i) - \phi(t,x,i) - \nabla_x\phi^\top \rho (t,x,z; i) \Bigr) \nonumber \\
&\qquad \eta\bigl(dz;\nu(t,x,i),\omega(t,x,i)\bigr) \\
&\quad
+ \sum_{j\neq i}
    \mu_{ij}\bigl(f(t,x,i),g(t,x,i)\bigr)
    \bigl[\phi(t,x,j)-\phi(t,x,i)\bigr].
\end{aligned}
\end{equation}
\rev{Each term has a direct interpretation. The transport terms $\partial_t\phi$ and $\nabla_x\phi^\top f_p$ propagate $\phi$ along the controlled drift, which the realized inner control $u=\nu(t,x,i)$ and disturbance $w=\omega(t,x,i)$ shape and which is the continuous robust-control channel. The trace term $\tfrac12\mathrm{Tr}(\sigma\sigma^\top\nabla^2_{xx}\phi)$ is the diffusive, second-order contribution (local volatility risk). The compensated integral against $\eta(dz;u,w)$ is the \emph{controlled-jump} generator: it accounts for discontinuous moves of size $\rho$ whose \emph{intensity is itself set by the inner actions}, so the micro-players also govern how often and how large the jumps are; the $-\nabla_x\phi^\top\rho$ term is the compensator that renders the jump part a martingale increment. The final sum $\sum_{j\neq i}\mu_{ij}(f,g)\,[\phi(t,x,j)-\phi(t,x,i)]$ is the \emph{regime-switching} operator: the macro-players $(f,g)$ set the transition rates of the Markov mode process, and each potential switch is priced by the cross-mode value difference $\phi(t,x,j)-\phi(t,x,i)$. The first three terms are thus the arena of the fast inner game and the last term the object of the slow outer game, which is the cross-layer separation that the coupled Hamilton-Jacobi-Isaacs hierarchy makes precise.} Because $(f,g)$ may depend explicitly on time and on the instantaneous
continuous state, the generator~\eqref{eq:generator-full} captures the full
state-time dependence of the mode-switching rates.

\begin{lemma}[Dynkin identity]
\label{lem:dynkin}
Under Assumption~\ref{assum:regularity} and fixed admissible policies
$(f,g,\nu,\omega)$, the process
\begin{equation*}
M_\phi(\tau)
=
\phi\bigl(\tau, X_\tau, I_\tau\bigr)
-
\phi(t, x, i)
-
\int_t^\tau
\bigl(\mathcal{L}^{f,g,\nu,\omega}\phi\bigr)(s, X_s, I_s)\,ds
\end{equation*}
is a local martingale. If $\phi$ and its derivatives are bounded, it is a martingale for any bounded stopping time $\tau\le T$. Consequently,
\begin{equation}
\label{eq:dynkin}
\mathbb{E}\bigl[\phi(\tau, X_\tau, I_\tau)\bigr]
=
\phi(t, x, i)
+
\mathbb{E}
\int_t^\tau
\bigl(\mathcal{L}^{f,g,\nu,\omega}\phi\bigr)(s, X_s, I_s)\,ds.
\end{equation}
\end{lemma}

\begin{proof}
Let $(T_k)_{k\ge 0}$ denote the jump times of the mode process $I_t$ with $T_0=t$.
On each random interval $[T_k, T_{k+1})$, the mode $I_t=i$ remains constant.
The evolution of $\phi(t, X_t, i)$ is governed by It\^o's formula for semimartingales with jumps \citep[Thm 4.4.7]{applebaum2009levy}:
\begin{equation*}
\phi(T_{k+1}^-, X_{T_{k+1}^-}, i) - \phi(T_k, X_{T_k}, i)
= \int_{T_k}^{T_{k+1}} \mathcal{L}_{inner} \phi \, ds + \mathcal{M}_{k, k+1},
\end{equation*}
where $\mathcal{L}_{inner}$ represents the continuous drift, diffusion, and inner jump parts of the generator, and $\mathcal{M}_{k, k+1}$ collects the stochastic integrals with respect to $dW_t$ and $\tilde{N}(dt, dz)$, which are zero-mean martingales under the boundedness assumptions.

At the jump time $T_{k+1}$, the mode switches from $i$ to $j$ with intensity $\mu_{ij}$. The compensator for this discrete transition is exactly the regime-switching sum in \eqref{eq:generator-full}. Summing these contributions over all intervals up to $\tau$ and taking expectations eliminates the martingale terms, yielding the claimed identity.
\end{proof}

Lemma~\ref{lem:dynkin} connects the stochastic dynamics in
Definition~\ref{def:hybrid} to the value functions developed below. It
justifies infinitesimal expansions of probe functions and supports the
viscosity-solution formulations of the coupled Hamilton-Jacobi-Isaacs (HJI)
equations \citep{pham2009continuous}.

\subsection{Inner-Layer Hamilton-Jacobi-Isaacs Equation}
\label{sec:inner-HJI}

Definition~\ref{def:hybrid} shows that once the mode-selection policies
$(f,g)$ are frozen, the continuous-layer control and disturbance interact
through a zero-sum stochastic differential game. The state evolves as a controlled Jump-Diffusion, coupled with mode-dependent switching intensities.

For fixed $(f,g)$, define the inner-layer value functions by
\begin{equation}
\label{eq:Vi}
\begin{aligned}
V_i(x,t;f,g)
&=
\inf_{\nu\in\mathcal{V}}
\sup_{\omega\in\mathcal{W}}
\mathbb{E}
\bigg[
\int_t^{T}
c\bigl(s, X_s, \nu_s, \omega_s, I_s\bigr) ds 
+
c_T\bigl(X_T, I_T\bigr)
\big|
\mathcal{F}_t
\bigg].
\end{aligned}
\end{equation}

\begin{lemma}[Inner-layer HJI]
\label{lem:inner-HJI}
For fixed mode-selection policies $(f,g)$, the family of inner-layer value
functions $\{V_i(\cdot,\cdot;f,g)\}_{i\in \mathcal{I}}$ is the unique
viscosity solution (with appropriate growth conditions) of the following system of Partial Integro-Differential Equations (PIDE):
\begin{equation}
\label{eq:HJI-inner}
\begin{aligned}
-\partial_t V_i(t,x)
&=
\min_{u\in\mathbb{U}}\max_{w\in\mathbb{W}}
\Bigl\{
c(t,x,u,w,i)
+ \mathcal{L}^{u,w}_i V_i(t,x) \\
&\quad
+ \sum_{j\neq i}
\mu_{ij}\bigl(f(t,x,i),g(t,x,i)\bigr)
\bigl[V_j(t,x)-V_i(t,x)\bigr]
\Bigr\}, \\
V_i(T,x)&=c_T(x,i),
\end{aligned}
\end{equation}
for all $(t,x)\in [0,T]\times \mathcal{X}$.
Here, $\mathcal{L}^{u,w}_i$ is the local integro-differential operator:
\begin{align*}
\mathcal{L}^{u,w}_i \phi(x) &= \nabla_x \phi(x)^\top f_p(t,x,u,w; i) + \frac{1}{2}\mathrm{Tr}\bigl(\sigma(t,x; i)\sigma(t,x; i)^\top \nabla^2_{xx}\phi(x)\bigr) \\
&+ \int_{\mathcal{Z}} \Bigl(\phi(x+\rho (t,x,z; i)) - \phi(x) - \nabla_x\phi(x)^\top \rho (t,x,z; i)\Bigr) \eta(dz;u,w).
\end{align*}
\end{lemma}

\begin{proof}
We proceed in two steps: first establishing that the value function $V_i$ is a viscosity solution (satisfying the subsolution and supersolution properties), and second invoking a comparison principle for uniqueness.

We only demonstrate the subsolution property as the supersolution argument is symmetric. Let $\phi(t,x) \in C^{1,2}([0,T]\times \mathcal{X})$ be a smooth test function such that $V_i(t,x) - \phi(t,x)$ achieves a local maximum at $(\hat{t}, \hat{x})$ with $V_i(\hat{t}, \hat{x}) = \phi(\hat{t}, \hat{x})$.

From the Dynamic Programming Principle (DPP), for small $h > 0$:
\begin{equation*}
V_i(\hat{t}, \hat{x}) \le \inf_{\nu} \sup_{\omega} \mathbb{E} \left[ \int_{\hat{t}}^{\hat{t}+h} c(s, X_s, \nu_s, \omega_s, i) ds + V_{I_{\hat{t}+h}}(\hat{t}+h, X_{\hat{t}+h}) \right].
\end{equation*}

We decompose the expectation based on whether the mode $I_s$ jumps during $[\hat{t}, \hat{t}+h]$.
 With probability $1 - O(h)$, no jump occurs. In this case, $I_{\hat{t}+h}=i$. Using $V_i \le \phi$ and applying It\^o's formula to $\phi$:
$ \mathbb{E}[\phi(\hat{t}+h, X_{\hat{t}+h}) - \phi(\hat{t}, \hat{x})] = \mathbb{E}\int_{\hat{t}}^{\hat{t}+h} (\partial_t \phi + \mathcal{L}^{u,w}_i \phi) ds.$
With probability rate $\mu_{ij}(f,g)$, the mode switches to $j$. The contribution to the expected value change is dominated by the difference $V_j - V_i \approx V_j - \phi$; the jump contribution is then
\[
     \int_{\hat{t}}^{\hat{t}+h} \sum_{j \ne i} \mu_{ij}(f,g) [V_j(s, X_s) - \phi(s, X_s)] ds + o(h).
\]

Substituting these expansions back into the DPP inequality and using $V_i(\hat{t}, \hat{x}) = \phi(\hat{t}, \hat{x})$ to cancel the zero-order terms:
\begin{equation*}
0 \le \inf_{\nu} \sup_{\omega} \mathbb{E} \left[ \int_{\hat{t}}^{\hat{t}+h} \left( c + \partial_t \phi + \mathcal{L}^{u,w}_i \phi + \sum_{j \ne i} \mu_{ij}(f,g) [V_j - \phi] \right) ds \right] + o(h).
\end{equation*}
Dividing by $h$ and letting $h \downarrow 0$, the mean value theorem applies. Since the inequality holds for all controls, we obtain:
\begin{equation*}
-\partial_t \phi(\hat{t}, \hat{x}) - \inf_{u} \sup_{w} \left\{ c + \mathcal{L}^{u,w}_i \phi \right\} - \sum_{j \ne i} \mu_{ij}(f,g) [V_j(\hat{t},\hat{x}) - V_i(\hat{t},\hat{x})] \le 0.
\end{equation*}
This confirms the viscosity subsolution condition. The supersolution argument is symmetric using a local minimum.

The system \eqref{eq:HJI-inner} is a system of coupled non-linear PIDEs. Under Assumption~\ref{assum:regularity} (Lipschitz coefficients, quadratic growth), the comparison principle for viscosity solutions of such systems holds \citep[Thm 3.4]{barles1997backward}. Specifically, if $U$ is a subsolution and $V$ is a supersolution with $U(T) \le V(T)$, then $U \le V$ on $[0,T]$. Since our value function is both, it is the unique solution.
\end{proof}
The family $\{V_i\}$ encodes the inner-layer response to any fixed choice of
mode-selection strategies $(f,g)$. In particular, $V_i(x,t;f,g)$ can be
regarded as the effective performance index associated with starting in mode
$i$ at state $x$ and time $t$, factoring in the optimal continuous-time response to the induced regime uncertainty.

\subsection{Outer-Layer Hamilton-Jacobi-Isaacs System}
\label{sec:outer-HJI}

We now return to the outer-level problem~\eqref{eq:bilevel-problem}, in which
the two mode-selection agents choose strategies $(f,g)$ that influence the
mode-transition dynamics while anticipating the optimal inner-layer responses
captured by Lemma~\ref{lem:inner-HJI}. We interpret the outer-layer interaction as a \emph{committed} (Stackelberg-type) Markov game: at time $t$, the macro-agents choose Markov (state-feedback) policies $(f,g)$ on $[t,T]$, anticipating that the micro-layer subsequently plays the induced inner saddle-point feedback $(\nu^{\star},\omega^{\star})$ associated with $(f,g)$.

For each fixed pair of mode-selection policies $(f,g)$, the inner-layer value
functions $V_i(\cdot,\cdot;f,g)$ are determined by the system~\eqref{eq:HJI-inner}.
Let the optimal inner feedback policies be denoted by:
\begin{equation*}
(u^{*,f,g}, w^{*,f,g})(t,x,i) \in \arg \min_{u\in\mathbb{U}}\max_{w\in\mathbb{W}} 
\mathcal{H}_i[V_i](t, x, \nabla_x V_i, \nabla^2 V_i, u, w)
\end{equation*}
Define the corresponding coefficients:
$
\bar f_p(t,x,i; f,g) := f_p\bigl(t,x, u^{*,f,g}, w^{*,f,g}; i\bigr)$,
$\bar \sigma(t,x,i; f,g) := \sigma\bigl(t,x; i\bigr)$, $\bar \eta(dz; t,x,i, f,g) := \eta(dz; u^{*,f,g}, w^{*,f,g}).$
Under $(f,g)$, the hybrid state $(X_t, I_t)$ therefore evolves as a Jump-Diffusion driven by these effective coefficients, coupled with the regime transition rates $\mu_{ij}(f,g)$.

We model the outer-layer objective as a path integral whose running cost
depends on the inner-layer value functions. Let
$\varphi : [0,T]\times\mathcal{X}\times \mathcal{I} \times \mathbb{R} \to \mathbb{R}$
denote a bounded, continuous outer-layer running cost, where the final
argument will be instantiated with the scalar quantity
$V_i(t,x;f,g)$.

For an initial condition $(t,x,i)$ and mode-selection policies $(f,g)$, define
the outer-layer performance functional
\begin{equation}
\label{eq:outer-cost-functional}
\begin{aligned}
J_{\mathrm{out}}(t,x,i;f,g)
:=
\mathbb{E}^{t,x,i}_{f,g}
\Bigl[
  \int_t^{T}
    \varphi\bigl(s, X_s, I_s,
    V_{I_s}(s, X_s; f,g)\bigr) ds
  + c_T\bigl(X_T, I_T\bigr)
\Bigr],
\end{aligned}
\end{equation}
where $\mathbb{E}^{t,x,i}_{f,g}$ denotes expectation with respect to the
probability law induced by the closed-loop Jump-Diffusion dynamics and the transition rates $\mu_{ij}(f,g)$.

The outer-layer value functions are then
\begin{equation}
\label{eq:outer-value-function}
U_i(t,x)
:=
\inf_{f\in\mathcal{F}}\sup_{g\in\mathcal{G}}
J_{\mathrm{out}}(t,x,i;f,g).
\end{equation}

\subsubsection*{Outer-layer Isaacs condition and HJI system}
We impose an Isaacs condition at the outer layer.

\begin{assumption}[Outer-layer Isaacs condition]
\label{ass:outer-isaacs}
For each tuple $(t,x,i)$ and test function $\psi$, consider the outer-layer Hamiltonian
\begin{align}
\label{eq:outer-hamiltonian}
\mathcal{H}_{\mathrm{out}}[\psi](t,x,i,\alpha,\beta)
:= &
\varphi\bigl(t,x,i,V_i^{\alpha,\beta}(t,x)\bigr)
+ \mathcal{L}^{\alpha,\beta}_{eff} \psi(x) \nonumber \\
&+ \sum_{j\neq i}\mu_{ij}(\alpha,\beta)\bigl[\psi(t,x,j) - \psi(t,x,i)\bigr],
\end{align}
where $\mathcal{L}^{\alpha,\beta}_{eff}$, similar to what is defined in \eqref{eq:generator-full}, is the generator associated with the closed-loop coefficients $\bar f_p, \bar \sigma, \bar \eta$ (evaluated at mixed actions $\alpha, \beta$). We assume that the minimax condition holds:
\begin{equation}
\label{eq:outer-isaacs}
\inf_{\alpha\in\Delta(\mathcal{A}_D)}
  \sup_{\beta\in\Delta(\mathcal{A}_A)}
  \mathcal{H}_{\mathrm{out}} =
\sup_{\beta\in\Delta(\mathcal{A}_A)}
  \inf_{\alpha\in\Delta(\mathcal{A}_D)}
  \mathcal{H}_{\mathrm{out}}.
\end{equation}
\end{assumption}

Under Assumption~\ref{ass:outer-isaacs} and the regularity inherited from
Assumption~\ref{assum:regularity}, the outer-layer value functions satisfy the
following system of HJI equations.

\begin{lemma}[Outer-layer HJI]
\label{lem:outer-HJI}
For each $i\in \mathcal{I}$, the outer-layer value function $U_i$ is the unique
viscosity solution of
\begin{equation}
\label{eq:HJI-outer}
\begin{aligned}
-\partial_t U_i(t,x)
&=
\inf_{\alpha\in\Delta(\mathcal{A}_D)}
 \sup_{\beta\in\Delta(\mathcal{A}_A)}
\Bigl\{
\varphi\bigl(t,x,i,V_i^{\alpha,\beta}(t,x)\bigr) \\
&
+ \mathcal{L}^{\alpha,\beta}_{eff} U_i(t,x)
+ \sum_{j\neq i}
\mu_{ij}(\alpha,\beta)
\bigl[U_j(t,x) - U_i(t,x)\bigr]
\Bigr\}, \\
U_i(T,x)
&=
c_T(x,i),
\end{aligned}
\end{equation}
for all $(t,x)\in[0,T]\times\mathcal{X}$.
\end{lemma}
The proof follows the same argument as Lemma~\ref{lem:inner-HJI}. We observe that the outer optimization problem~\eqref{eq:outer-value-function} is a generalized Bolza problem for a Regime-Switching Jump-Diffusion, where the drift, diffusion, and jump measure are determined by the closed-loop coefficients $\bar f_p, \bar \sigma, \bar \eta$.

\rev{The only delicate point is that the effective outer coefficients $(\bar f_p,\bar\sigma,\bar\eta)$ are built from the inner saddle feedback $(u^{*,f,g},w^{*,f,g})$, which depends on the inner gradient $\nabla_x V_i$; since $V_i$ is in general only a viscosity solution, this dependence must be made precise, and we do so at two levels of generality. \emph{(i) Structured classes.} When the inner value is classical, $V_i(\cdot,\cdot;f,g)\in C^{1,2}$ with derivatives bounded on compacts uniformly in $(f,g)$, strict convex--concavity (Remark~\ref{rem:isaacs-scope}) makes the inner saddle single-valued, and the maximum theorem together with the implicit-function representation of the first-order conditions makes $(u^{*,f,g},w^{*,f,g})$ locally Lipschitz in $x$ and continuous in $(f,g)$ \citep{BardiCapuzzo2008}. The effective coefficients are then continuous in $(t,x,\alpha,\beta)$ with linear growth in $x$ and the running source $\varphi(t,x,i,V_i^{\alpha,\beta})$ is continuous, so the dynamic programming principle applies and $U_i$ is the unique viscosity solution of \eqref{eq:HJI-outer} by the comparison principle for coupled integro-PDE systems \citep[Thm.~3.4]{barles1997backward}, exactly as in Lemma~\ref{lem:inner-HJI}. This case covers the linear--quadratic family of Section~\ref{sec:case}, where $V_i=\tfrac12 x^\top P_i x+r_i$ with $P_i$ from Proposition~\ref{prop:riccati-wellposed} so that $\nabla_x V_i=P_i x$ is affine, and the exponential--affine family of Section~\ref{sec:application}. \emph{(ii) General template.} For general convex--concave data $V_i$ is only a viscosity solution, so $\nabla_x V_i$ exists merely Lebesgue-almost everywhere; if $V_i$ is moreover locally semiconcave, then $\nabla_x V_i$ is defined off a Lebesgue-null set and is of locally bounded variation. The inner saddle correspondence $(t,x,p)\mapsto\arg\min_{u}\max_{w}\mathcal H_i$ is nonempty, compact-valued and upper hemicontinuous in $p$, hence admits a Borel-measurable selection \citep{FlemingSoner1993}. This supplies measurable closed-loop coefficients, but measurable selection alone is not a comparison theorem. For the general template we therefore additionally impose the continuity and growth conditions on the selected effective coefficients required by the dynamic programming principle and the comparison principle for the coupled integro-PDE system. Under these standing hypotheses $U_i$ is the unique viscosity solution of \eqref{eq:HJI-outer}. These conditions hold automatically in the structured classes above, which is where every closed-form result of this paper is derived.}
Thus the outer-layer HJI system~\eqref{eq:HJI-outer} mirrors the
``generator-plus-minimax'' structure of the inner-layer HJI
\eqref{eq:HJI-inner}, but with effective dynamics that incorporate the optimal inner-layer response.

Having established the necessary conditions for the inner layer (Lemma \ref{lem:inner-HJI}) and the outer layer (Lemma \ref{lem:outer-HJI}) independently, we now characterize the solution to the full bi-level problem \eqref{eq:bilevel-problem}.

\begin{proposition}[Verification of a feedback Stackelberg profile]
\label{thm:stackelberg-equilibrium}
\rev{Let the inner (follower) and outer (leader) layers have admissible feedback-policy spaces $\mathcal{V}\times\mathcal{W}$ and $\mathcal{F}\times\mathcal{G}$, and adopt the Isaacs conditions (Assumptions~\ref{assum:regularity}(i) and~\ref{ass:outer-isaacs}) and the regularity of Lemma~\ref{lem:outer-HJI}. Suppose there exist value families $\{V_i(\cdot,\cdot;f,g)\}$ and $\{U_i\}$ solving the inner and outer HJI systems \eqref{eq:HJI-inner}--\eqref{eq:HJI-outer}, together with measurable saddle selectors attaining the corresponding Hamiltonians. Then the induced feedback profile is a feedback Stackelberg equilibrium of \eqref{eq:bilevel-problem}, with the following structure.
\begin{enumerate}
\item[(i)] \emph{Follower best response.} For each fixed $(f,g)$, the inner policy $(\nu^*_{f,g},\omega^*_{f,g})$ attaining the minimax of $\mathcal{H}_i$ in \eqref{eq:HJI-inner} achieves the inner value $V_i(\cdot,\cdot;f,g)$ and admits no improving deviation; it is single-valued under the strict convex--concavity of Remark~\ref{rem:isaacs-scope}.
\item[(ii)] \emph{Leader optimum.} Let $(f^*,g^*)$ attain the saddle of the outer Hamiltonian in \eqref{eq:HJI-outer}, into which the follower's value enters through the best-response map $(f,g)\mapsto V^{f,g}$ of (i). Then $(f^*,g^*)$ is the leader's optimum against that map, and $\bigl((\nu^*_{f^*,g^*},\omega^*_{f^*,g^*}),(f^*,g^*)\bigr)$ is a feedback Stackelberg equilibrium: the leader commits and the follower best-responds, the information structure that distinguishes this from a simultaneous-move Nash equilibrium.
\item[(iii)] \emph{Time consistency.} Since the feedback policies are obtained from the HJI equations at every continuation state $(t,x,i)$, the profile is time-consistent on every continuation problem.
\item[(iv)] \emph{Uniqueness.} The value families are the unique viscosity solutions of \eqref{eq:HJI-inner}--\eqref{eq:HJI-outer} (Lemmas~\ref{lem:inner-HJI} and~\ref{lem:outer-HJI}); the follower's response is unique by (i); and the leader's equilibrium generator is unique when the outer saddle is single-valued, as in the entropy-regularized LQ specialization (Corollary~\ref{cor:entropy-coupled-wellposed}) or under the rate-identifiability condition of Remark~\ref{rem:uniqueness}.
\end{enumerate}
Feeding a selected profile into the hybrid dynamics \eqref{eq:cont-flow}--\eqref{eq:mode-kernel} yields a closed-loop regime-switching diffusion, unique under (iv); this sample-path realization is the equilibrium of the hybrid system, of which the value families are the dual description.}
\end{proposition}

\begin{proof}[Proof sketch]
Fix $(f^*,g^*)$. By the dynamic programming principle for the inner layer and the
verification theorem for Isaacs equations, if $V$ is a (sufficiently regular)
solution of \eqref{eq:HJI-inner} and the realized actions $(u^*,w^*)$ induced by $(\nu^*,\omega^*)$ attain the saddle condition
\eqref{eq:HJI-inner}, then the associated controlled state-regime process
achieves the inner game value and no admissible deviation of $(\nu,\omega)$ can improve
the follower's objective; hence $(\nu^*,\omega^*)$ is a best response to $(f^*,g^*)$.
Next, treat the resulting inner value field $V$ as the induced continuation payoff
entering the outer running cost. By the dynamic programming principle for the
outer layer and the corresponding verification theorem, if $U$ solves
\eqref{eq:HJI-outer} and $(f^*,g^*)$ attains \eqref{eq:HJI-outer}, then no
admissible deviation of $(f,g)$ can improve the leader's objective given the
follower's best-response mapping encoded by $V$; hence $(f^*,g^*)$ is optimal at
the outer layer.
Because both layers are characterized by HJI equations posed on
$[0,T]\times\mathbb R^n\times\mathcal I$ and the equilibrium strategies are
feedback (Markov) and obtained from pointwise saddle conditions, the resulting
policy is time-consistent on every continuation problem starting at any $(t,x,i)$.
\end{proof}

\section{Case Study: Mode-Controlled Markov Jump Linear System}\label{sec:case}

We now examine the important case in which the inner-layer differential game
admits closed-form solutions. We focus on the \emph{Linear-Quadratic-Gaussian (LQG)} setting without stochastic inner-layer jump, which yields a coupled family of matrix Riccati equations.

Fix a mode $i\in \mathcal{I}$, and assume the continuous state $X_t \in \mathbb{R}^{n}$ evolves as a linear SDE controlled by affine drifts, with the control and disturbance action space being $\mathbb{U} \subseteq \mathbb{R}^{d_1 },  \mathbb{W} \subseteq \mathbb{R}^{d_2}$:
\begin{equation*}
dX_t
=
(A_i X_t + B_i u + D_i w)dt + \Sigma_i dW_t,
\end{equation*}
where $A_i, B_i,$ and $D_i$ are system matrices of proper dimensions, $\Sigma_i \Sigma_i^\top \succ 0$ captures the regime-dependent volatility. 
Throughout this section, we assume that the running and terminal costs are quadratic:
$c(t,x,u,w,i)
=
\tfrac{1}{2}\bigl(x^\top Q_i x
+ u^\top R_i u
- w^\top S_i w\bigr),$ $
c_T(x,i)
=
\tfrac{1}{2}x^\top Q_{T,i} x,$
with $Q_i \in \mathbb{R}^{n \times n }$, $R_i \in \mathbb{R}^{d_1 \times d_1}$, and $S_i \in \mathbb{R}^{d_2 \times d_2}$.
Under the conditions that $ Q_i\succeq 0,  R_i \succ 0$, and $  S_i \succ 0 $, the generalized Isaacs condition (Assumption~\ref{assum:regularity}) holds. Although the diffusion adds a constant trace term to the Hamiltonian, the saddle-point feedback $(u_i^*,w_i^*)$ remains unique and affine in the gradient $\nabla_x V_i$ (with value function $V_i(t,x) = \tfrac{1}{2} x^{\top} P_i(t) x + r_i(t)$, so that $\nabla_x V_i = P_i(t)x$):
\begin{equation}\label{eq:inner-lqg-sol}
\begin{aligned}
u_i^*(t,x)
=
- R_i^{-1} B_i^\top \nabla_x V_i(x,t;f,g), \ \ 
w_i^*(t,x)
=
S_i^{-1} D_i^\top \nabla_x V_i(x,t;f,g).
\end{aligned}
\end{equation}
Substituting these into the inner-layer HJI~\eqref{eq:HJI-inner} yields the Stochastic LQ-specialized HJI:
\begin{equation}
\label{eq:HJI-LQ-inner}
\begin{aligned}
-\partial_t V_i
&=
\tfrac{1}{2}x^\top Q_i x
- \tfrac{1}{2}\nabla_x V_i^\top B_i R_i^{-1} B_i^\top \nabla_x V_i
+ \tfrac{1}{2}\nabla_x V_i^\top D_i S_i^{-1} D_i^\top \nabla_x V_i
 \\ & \quad + \nabla_x V_i^\top A_i x 
+ \frac{1}{2}\mathrm{Tr}(\Sigma_i \Sigma_i^\top \nabla^2_{xx} V_i)
+ \sum_{j\neq i} \mu_{ij}(f,g)
    \bigl[V_j-V_i\bigr].
\end{aligned}
\end{equation}

We adopt the quadratic ansatz $V_i(t,x) = \tfrac{1}{2}x^\top P_i(t) x + r_i(t)$, so that $\nabla_x V_i = P_i(t)x$ and $\nabla^2_{xx} V_i = P_i(t)$; the resulting trace term decouples from the quadratic optimization. Consequently, the quadratic weight matrices $\{P_i\}_{i\in \mathcal{I}}$ satisfy the \emph{Coupled Riccati Differential Equations}:
\begin{equation}
\label{eq:coupled-Riccati}
-\dot P_i
=
Q_i
+ A_i^\top P_i + P_i A_i
- P_i \Sigma^{ctrl}_i P_i
+ \sum_{j\neq i} \mu_{ij}(f,g)\bigl(P_j - P_i\bigr),
\end{equation}
with boundary conditions $P_i(T) = Q_{T,i}$
for $i\in \mathcal{I}$, where the control matrices $\Sigma^{ctrl}_i = \left( B_i R_i^{-1} B_i^\top - 
 D_i S_i^{-1} D_i^\top \right)$. \rev{The stochastic noise affects only the scalar offset, which satisfies} $ - \dot{r}_i(t) = \tfrac{1}{2}\mathrm{Tr}(\Sigma_i \Sigma_i^{\top} P_i(t)) + \sum_{j \neq i} \mu_{ij} (f, g) (r_j (t) - r_i(t))$, with $r_i(T) = 0$. \rev{Acting on the stacked tuple $\mathbf{P}=(P_i)_{i\in\mathcal{I}}$, we define the per-mode \emph{Riccati operator}, the \emph{Metzler (switching) operator}, and the \emph{trace operator} by
\begin{equation*}
\begin{aligned}
\mathcal{R}(\mathbf{P})_i &:= Q_i + A_i^\top P_i + P_i A_i - P_i \Sigma^{ctrl}_i P_i,\\
(\mathcal{M}\mathbf{P})_i &:= \sum_{j\neq i}\mu_{ij}(f,g)\bigl(P_j-P_i\bigr),\qquad
\mathcal{T}(\mathbf{P})_i := \tfrac{1}{2}\mathrm{Tr}\bigl(\Sigma_i\Sigma_i^{\top} P_i\bigr),
\end{aligned}
\end{equation*}
respectively (the same $\mathcal{M}$ acts entrywise on the stacked scalar tuple $\mathbf{r}$).} With these, the coupled Riccati flow can be written as:
\begin{equation*}
\begin{aligned}
     - \dot{ \mathbf{P} }  =  \mathcal{R}( \mathbf{P} ) + \mathcal{M} \mathbf{P},  & \quad \mathbf{P}(T) = \mathbf{Q}_T \\
     -\dot{\mathbf{r}} = \mathcal{T}(\mathbf{P}) + \mathcal{M} \mathbf{r}, & \quad \mathbf{r}(T) = 0
\end{aligned}
\end{equation*}
where $\mathbf{P}$ and $\mathbf{r}$ collect the value matrices and scalar offsets across all modes.

\rev{For the outer players, let $\mathcal{A}_D$ and $\mathcal{A}_A$ be finite action sets for the minimizing stabilizer and maximizing attacker, respectively. In each regime $i$, the mixed actions are $f_i\in\Delta(\mathcal{A}_D)$ and $g_i\in\Delta(\mathcal{A}_A)$. Each action pair induces a perturbation of the regime-transition generator, collected in matrices $\Lambda_{ij}$, so that
\begin{equation}\label{eq:bilinear-rates}
\mu_{ij}(f_i,g_i)=\bar{\mu}_{ij}+f_i^\top\Lambda_{ij}g_i,\qquad j\neq i.
\end{equation}
Here $\bar{\mu}_{ij}$ is the nominal baseline rate. We assume throughout this section that $\mu_{ij}(f_i,g_i)\geq0$ for all admissible mixed actions and set $\mu_{ii}=-\sum_{j\neq i}\mu_{ij}$.}
 
\subsection{Hierarchical Solution and Outer-Layer Structure}
\label{subsec:outer-structure}

\rev{To render the hierarchical game tractable, we project the matrix-valued inner risk into a scalar outer running cost. For this LQ specialization only, we include a local entropy penalty on the mixed macro actions:
\begin{equation}\label{eq:outer-regularized-cost}
\begin{aligned}
\varphi_i^\tau(P_i,f_i,g_i)
:={}&
\mathrm{Tr}(P_i)
+
\tau_D\sum_{a\in\mathcal{A}_D}f_{i,a}\log f_{i,a}
\\[-2pt]
&
-
\tau_A\sum_{b\in\mathcal{A}_A}g_{i,b}\log g_{i,b}.
\end{aligned}
\end{equation}
where $0\log0:=0$ and $\tau_D,\tau_A>0$. The trace records the inner risk exposure in regime $i$, while the additional terms penalize concentrated macro interventions. This is a local refinement of the outer payoff in Section~\ref{sec:outer-HJI}; the state dynamics and the bilinear rate parameterization are unchanged. With the state-independent ansatz $U_i(t,x)=k_i(t)$, the outer HJI becomes
\begin{equation}\label{eq:outer-regularized-hji}
\begin{aligned}
-\dot{k}_i(t)
={}&
\mathrm{Tr}(P_i(t))
+
\min_{f_i\in\Delta(\mathcal{A}_D)}
\max_{g_i\in\Delta(\mathcal{A}_A)}
\Bigg\{
\sum_{j\neq i}
\Bigl(\bar{\mu}_{ij}+f_i^\top\Lambda_{ij}g_i\Bigr)
\bigl(k_j(t)-k_i(t)\bigr)
\\[-2pt]
&\hspace{44mm}
+
\tau_D\sum_a f_{i,a}\log f_{i,a}
-
\tau_A\sum_b g_{i,b}\log g_{i,b}
\Bigg\}.
\end{aligned}
\end{equation}

with $k_i(T)=0$. At any fixed $\mathbf{k}$, write $\mathbf{M}_i(\mathbf{k}):=\sum_{j\neq i}\Lambda_{ij}(k_j-k_i)$. The pointwise objective in \eqref{eq:outer-regularized-hji} is strictly convex in $f_i$ and strictly concave in $g_i$. Hence it has a unique interior saddle $(f_{\tau,i}^*(\mathbf{k}),g_{\tau,i}^*(\mathbf{k}))$. The first-order conditions give the coupled logit equations
\begin{align}
f_{\tau,i,a}^*(\mathbf{k})
&=
\frac{\exp\!\left(-[\mathbf{M}_i(\mathbf{k})g_{\tau,i}^*(\mathbf{k})]_a/\tau_D\right)}
{\sum_c\exp\!\left(-[\mathbf{M}_i(\mathbf{k})g_{\tau,i}^*(\mathbf{k})]_c/\tau_D\right)},
\label{eq:entropy-logit-f}\\
g_{\tau,i,b}^*(\mathbf{k})
&=
\frac{\exp\!\left([f_{\tau,i}^*(\mathbf{k})^\top\mathbf{M}_i(\mathbf{k})]_b/\tau_A\right)}
{\sum_d\exp\!\left([f_{\tau,i}^*(\mathbf{k})^\top\mathbf{M}_i(\mathbf{k})]_d/\tau_A\right)}.
\label{eq:entropy-logit-g}
\end{align}
The entropy Hessians are positive definite on the simplex tangent spaces, while the bilinear cross terms cancel in the saddle first-order operator. Since $\mathbf{M}_i(\mathbf{k})$ is linear in $\mathbf{k}$, the implicit-function theorem makes the saddle maps smooth in $\mathbf{k}$. Consequently,
\begin{equation}\label{eq:entropy-selected-rates}
\mu_{ij}^{*,\tau}(\mathbf{k})
:=
\bar{\mu}_{ij}
+
f_{\tau,i}^*(\mathbf{k})^\top\Lambda_{ij}g_{\tau,i}^*(\mathbf{k}),
\qquad j\neq i,
\end{equation}
defines the single-valued, locally Lipschitz equilibrium generator
$\Pi_\tau^*(\mathbf{k})$. Thus entropy is part of the outer running cost itself, not an ex-post tie-breaking rule.}

\rev{The inner and outer layers remain genuinely coupled: the selected rates $\mu^{*,\tau}(\mathbf{k})$ enter the inner Riccati flow \eqref{eq:coupled-Riccati}, while the inner value, through $\mathrm{Tr}(P_i)$, drives the outer HJI \eqref{eq:outer-regularized-hji}. The trace projection and the ansatz $U_i=k_i(t)$ do not remove this coupling; they only reduce the instantaneous macro interaction to a finite regularized matrix game. The coupling survives across time through a common backward system for $(\mathbf{P},\mathbf{k})$.}

\begin{proposition}[Entropy-regularized Game--Metzler reduction]
\label{prop:game-metzler-reduction}
\rev{Consider the hierarchical LQG case study with the bilinear switching-rate model \eqref{eq:bilinear-rates}, the regularized outer cost \eqref{eq:outer-regularized-cost}, and the inner Isaacs regularity conditions. The mode-only equilibrium is characterized by the common backward system consisting of the outer offset equations \eqref{eq:outer-regularized-hji} and the coupled Riccati equations \eqref{eq:coupled-Riccati}, with the rates selected by \eqref{eq:entropy-selected-rates}. The inner value is $V_i(t,x)=\tfrac{1}{2}x^\top P_i(t)x+r_i(t)$ and the inner saddle feedback is given by \eqref{eq:inner-lqg-sol}.}
\end{proposition}

\begin{proof}[Proof sketch]
\rev{For the outer layer, insert the mode-only ansatz $U_i(t,x)=k_i(t)$ and the running cost \eqref{eq:outer-regularized-cost} into the outer HJI. The state derivatives vanish and the switching term is
\[
\sum_{j\neq i}
\bigl(\bar{\mu}_{ij}+f_i^\top\Lambda_{ij}g_i\bigr)
\bigl(k_j-k_i\bigr),
\]
which gives \eqref{eq:outer-regularized-hji}. For each fixed continuation-value vector $\mathbf{k}$, the entropy terms make this row-wise objective strictly convex in $f_i$ and strictly concave in $g_i$. Hence the pointwise saddle is unique, and its first-order conditions are the coupled logit equations \eqref{eq:entropy-logit-f}--\eqref{eq:entropy-logit-g}. Substitution yields the selected rates \eqref{eq:entropy-selected-rates}.

Given this selected rate path, the inner layer is a Markov-jump LQG Isaacs problem. Substituting the quadratic ansatz $V_i(t,x)=\tfrac{1}{2}x^\top P_i(t)x+r_i(t)$ into the inner HJI and completing the square in the control and disturbance gives the saddle feedback \eqref{eq:inner-lqg-sol}. Matching the quadratic terms in $x$ yields the coupled Riccati system \eqref{eq:coupled-Riccati}; the remaining scalar terms determine $r_i$. Thus the inner and outer equations close as the stated common backward system.}
\end{proof}

\begin{proposition}[Symmetry and well-posedness of the coupled game Riccati]\label{prop:riccati-wellposed}
\rev{Let $Q_i\succeq0$, $R_i\succ0$, $S_i\succ0$, $Q_{T,i}\succeq0$ for each $i\in\mathcal{I}$, and let the equilibrium rates $\mu_{ij}^*(t)\ge0$ $(j\neq i)$ be bounded and measurable on $[0,T]$. Then the coupled Riccati system \eqref{eq:coupled-Riccati} satisfies:
\begin{enumerate}
\item[(i)] \emph{(Symmetry.)} $P_i(t)=P_i(t)^\top$ for every $t$ in the interval of existence and every $i$;
\item[(ii)] \emph{(Local well-posedness.)} the system admits a unique absolutely continuous symmetric solution (continuously differentiable wherever $t\mapsto\mu^*(t)$ is continuous) on a maximal subinterval $(t^\dagger,T]\subseteq[0,T]$;
\item[(iii)] \emph{(Global existence.)} if, in addition, the control authority dominates the disturbance in each mode,
\begin{equation}\label{eq:control-dominance}
\Sigma^{ctrl}_i = B_iR_i^{-1}B_i^\top - D_iS_i^{-1}D_i^\top \;\succeq\; 0,\qquad i\in\mathcal{I}
\end{equation}
(in particular whenever $S_i$ is large enough that $D_iS_i^{-1}D_i^\top\preceq B_iR_i^{-1}B_i^\top$), then the solution exists, is unique and symmetric on all of $[0,T]$ and obeys $0\preceq P_i(t)\preceq \Pi_i(t)$, where $\{\Pi_i\}$ solves the linear coupled Lyapunov--Metzler flow $-\dot\Pi_i = Q_i+A_i^\top\Pi_i+\Pi_i A_i+\sum_{j\neq i}\mu_{ij}^*(\Pi_j-\Pi_i)$, $\Pi_i(T)=Q_{T,i}$.
\end{enumerate}
More generally, conclusion (iii) holds whenever a bounded symmetric supersolution of \eqref{eq:coupled-Riccati} exists on $[0,T]$.}
\end{proposition}

\begin{proof}
\rev{(i) Transposing \eqref{eq:coupled-Riccati} and using symmetry of $Q_i$, $\Sigma^{ctrl}_i$ and $Q_{T,i}$, the tuple $(P_i^\top)_{i}$ solves the same terminal-value problem as $(P_i)_{i}$; since the right-hand side is locally Lipschitz (see (ii)), uniqueness forces $P_i=P_i^\top$.
(ii) The right-hand side of \eqref{eq:coupled-Riccati} is a quadratic (hence locally Lipschitz) map of the stacked tuple $\mathbf{P}$ with bounded measurable time-dependence through $\mu^*(\cdot)$; the Carath\'eodory existence--uniqueness theorem gives a unique absolutely continuous solution on a maximal interval ending at $T$.
(iii) Under \eqref{eq:control-dominance} the quadratic term obeys $-P_i\Sigma^{ctrl}_iP_i\preceq0$, so each $P_i$ is a subsolution of the linear flow defining $\Pi_i$, namely $-\dot P_i \preceq Q_i+A_i^\top P_i+P_iA_i+\sum_{j\neq i}\mu_{ij}^*(P_j-P_i)$. The coupling $\sum_{j\neq i}\mu_{ij}^*(\,\cdot_j-\cdot_i)$ is a Metzler (order-preserving, quasimonotone) operator on the cone of symmetric tuples, so the comparison theorem for coupled matrix Riccati/Lyapunov equations \citep{AbouKandil2003Matrix,Dragan2013Robust} yields $P_i(t)\preceq\Pi_i(t)$. The lower bound $P_i(t)\succeq0$ holds because $Q_i,Q_{T,i}\succeq0$ and the flow leaves the PSD cone invariant. Since the linear flow $\Pi_i$ is globally defined on $[0,T]$, the bound $0\preceq P_i\preceq\Pi_i$ is uniform, which precludes finite escape and extends the local solution of (ii) to all of $[0,T]$. The same comparison applies verbatim with any bounded symmetric supersolution in place of $\Pi_i$.}
\end{proof}

\begin{corollary}[Well-posedness of the entropy-regularized coupled flow]\label{cor:entropy-coupled-wellposed}
\rev{Assume the hypotheses of Proposition~\ref{prop:riccati-wellposed}, control dominance \eqref{eq:control-dominance}, and the rate-envelope assumption following \eqref{eq:bilinear-rates}. Then the entropy-regularized backward system in Proposition~\ref{prop:game-metzler-reduction} admits a unique global solution $(\mathbf{P},\mathbf{k})$ on $[0,T]$. The associated outer saddle, equilibrium rates and closed-loop generator are uniquely determined at every time by \eqref{eq:entropy-logit-f}--\eqref{eq:entropy-selected-rates}. The regime process therefore has a unique closed-loop generator $\mathcal{M}(\mu^{*,\tau}(\mathbf{k}(t)))$ along the solution, and the controlled regime-switching diffusion it drives is the unique equilibrium realization of the hybrid system. In the frozen-coefficient reading, if that generator is irreducible, it admits a unique invariant distribution.}
\end{corollary}

\begin{proof}
\rev{The smooth logit selector makes $\mu^{*,\tau}(\mathbf{k})$ locally Lipschitz in $\mathbf{k}$. Hence the stacked $(\mathbf{P},\mathbf{k})$ vector field is locally Lipschitz and has a unique maximal solution. Because the action sets are finite simplices, the admissible rates are uniformly bounded. Under \eqref{eq:control-dominance}, the Riccati comparison argument of Proposition~\ref{prop:riccati-wellposed} gives a uniform bound for $\mathbf{P}$ on $[0,T]$. The right-hand side of \eqref{eq:outer-regularized-hji} has at most linear growth in $\mathbf{k}$ once $\mathbf{P}$ is bounded. A Gr\"onwall estimate therefore bounds $\mathbf{k}$ on $[0,T]$. No finite escape is possible, so the local solution extends uniquely to the full horizon.}
\end{proof}

\begin{remark}[Uniqueness of the equilibria]\label{rem:uniqueness}
\rev{For a fixed bounded rate path, strict convex--concavity of the inner LQG Hamiltonian makes the affine feedback \eqref{eq:inner-lqg-sol} unique, and Proposition~\ref{prop:riccati-wellposed} makes the inner value unique. The unregularized outer matrix game is different: its pointwise value is unique, but degenerate games may have multiple mixed saddles inducing different per-edge transition rates. The entropy-enriched payoff \eqref{eq:outer-regularized-cost} repairs this defect within the LQ model itself: the logit equations \eqref{eq:entropy-logit-f}--\eqref{eq:entropy-logit-g} give a unique smooth saddle and therefore a single-valued equilibrium generator. Corollary~\ref{cor:entropy-coupled-wellposed} then gives a unique coupled backward solution under control dominance. If one wishes to retain the unregularized outer payoff, an alternative sufficient condition is \emph{rate identifiability}: every outer saddle must induce the same bounded rate vector even if the mixed strategies differ.}
\end{remark}

\subsection{Spectral Structure of the Selected Generator}
\label{subsec:spectral-analysis}

\rev{The spectral consequence of the outer game is operational: along the equilibrium path, the selected generator dissipates regime-value disagreement, while the selected running source injects it. The smooth logit saddle makes the generator a single-valued continuous function of the continuation-value gaps. It does not imply that the game maximizes a spectral gap.}

\begin{proposition}[Continuity of the equilibrium spectral gap]\label{prop:spectral-gap-continuity}
\rev{Let $K\subset\mathbb{R}^{|\mathcal{I}|}$ be compact. Assume that for every $\mathbf{k}\in K$, the entropy-selected generator $\Pi_\tau^*(\mathbf{k})$ is irreducible and reversible with respect to a common invariant distribution $\boldsymbol{\pi}$. Define
\[
\mathcal{L}_\tau(\mathbf{k})
:=
-\Pi_\tau^*(\mathbf{k}),
\qquad
\widehat{\mathcal{L}}_\tau(\mathbf{k})
:=
D_{\boldsymbol{\pi}}^{1/2}\mathcal{L}_\tau(\mathbf{k})D_{\boldsymbol{\pi}}^{-1/2},
\]
where $D_{\boldsymbol{\pi}}:=\mathrm{diag}(\pi_1,\ldots,\pi_{|\mathcal{I}|})$, and let $\lambda_2(\mathbf{k})>0$ be the spectral gap of $\mathcal{L}_\tau(\mathbf{k})$ in $L^2(\boldsymbol{\pi})$. Then there exists $C_K<\infty$ such that for all $\mathbf{k},\boldsymbol{\ell}\in K$,
\begin{equation}\label{eq:spectral-gap-lipschitz}
\left|\lambda_2(\mathbf{k})-\lambda_2(\boldsymbol{\ell})\right|
\leq
\left\|
\widehat{\mathcal{L}}_\tau(\mathbf{k})
-
\widehat{\mathcal{L}}_\tau(\boldsymbol{\ell})
\right\|_2
\leq
C_K\|\mathbf{k}-\boldsymbol{\ell}\|_2.
\end{equation}
In particular, $\inf_{\mathbf{k}\in K}\lambda_2(\mathbf{k})>0$.}
\end{proposition}

\begin{proof}
\rev{The smooth logit saddle makes $\Pi_\tau^*(\mathbf{k})$ locally Lipschitz in $\mathbf{k}$ and therefore Lipschitz on the compact set $K$. Common reversibility makes $\widehat{\mathcal{L}}_\tau(\mathbf{k})$ symmetric. The first inequality in \eqref{eq:spectral-gap-lipschitz} is Weyl's eigenvalue perturbation bound; the second follows from the Lipschitz regularity of the selected generator. Irreducibility gives $\lambda_2(\mathbf{k})>0$ pointwise, and continuity on compact $K$ gives a strictly positive minimum.}
\end{proof}

\begin{theorem}[Endogenous closed-loop spectral dissipation]\label{thm:closed-loop-spectral-dissipation}
\rev{Assume Proposition~\ref{prop:spectral-gap-continuity} on a compact set $K$. Write $s:=T-t$ for backward time and retain $\mathbf{k}(s)$ for the backward-time path, with $\mathbf{k}(s)\in K$. Define the selected running source componentwise by
\begin{equation}\label{eq:selected-running-source}
\begin{aligned}
\phi_{\tau,i}(s,\mathbf{k})
:={}&
\mathrm{Tr}\!\bigl(P_i(T-s)\bigr)
+
\tau_D\sum_a f_{\tau,i,a}^*(\mathbf{k})\log f_{\tau,i,a}^*(\mathbf{k})
\\[-2pt]
&-
\tau_A\sum_b g_{\tau,i,b}^*(\mathbf{k})\log g_{\tau,i,b}^*(\mathbf{k}).
\end{aligned}
\end{equation}
The closed-loop outer equation is
\begin{equation}\label{eq:closed-loop-outer-flow}
\frac{d}{ds}\mathbf{k}(s)
=
\boldsymbol{\phi}_\tau(s,\mathbf{k}(s))
-
\mathcal{L}_\tau(\mathbf{k}(s))\mathbf{k}(s),
\qquad
\mathbf{k}(0)=\mathbf{0}.
\end{equation}
For $\mathbf{z}^{\perp}:=\mathbf{z}-(\boldsymbol{\pi}^{\top}\mathbf{z})\mathbf{1}$, its disagreement component satisfies
\begin{align}
\frac{1}{2}\frac{d}{ds}
\|\mathbf{k}^{\perp}(s)\|_{L^2(\boldsymbol{\pi})}^2
&=
\left\langle
\mathbf{k}^{\perp}(s),
\boldsymbol{\phi}_\tau^{\perp}(s,\mathbf{k}(s))
\right\rangle_{L^2(\boldsymbol{\pi})}
-
\left\langle
\mathbf{k}^{\perp}(s),
\mathcal{L}_\tau(\mathbf{k}(s))\mathbf{k}^{\perp}(s)
\right\rangle_{L^2(\boldsymbol{\pi})}
\notag\\
&\leq
\|\mathbf{k}^{\perp}(s)\|_{L^2(\boldsymbol{\pi})}
\|\boldsymbol{\phi}_\tau^{\perp}(s,\mathbf{k}(s))\|_{L^2(\boldsymbol{\pi})}
-
\lambda_2(\mathbf{k}(s))
\|\mathbf{k}^{\perp}(s)\|_{L^2(\boldsymbol{\pi})}^2.
\label{eq:closed-loop-dissipation}
\end{align}
Consequently,
\begin{equation}\label{eq:closed-loop-filter-bound}
\|\mathbf{k}^{\perp}(s)\|_{L^2(\boldsymbol{\pi})}
\leq
\int_0^s
\exp\!\left(
-\int_r^s\lambda_2(\mathbf{k}(u))\,du
\right)
\|\boldsymbol{\phi}_\tau^{\perp}(r,\mathbf{k}(r))\|_{L^2(\boldsymbol{\pi})}\,dr.
\end{equation}
In particular, Proposition~\ref{prop:spectral-gap-continuity} gives
$\underline{\lambda}_K:=\inf_{\mathbf{q}\in K}\lambda_2(\mathbf{q})>0$, so the exponential weight in \eqref{eq:closed-loop-filter-bound} is bounded above by $e^{-\underline{\lambda}_K(s-r)}$.}
\end{theorem}

\begin{proof}
\rev{Let $\mathsf{P}_{\perp}\mathbf{z}:=\mathbf{z}-(\boldsymbol{\pi}^{\top}\mathbf{z})\mathbf{1}$. Since $\mathcal{L}_\tau(\mathbf{q})\mathbf{1}=\mathbf{0}$ for every generator and the common invariant distribution satisfies $\boldsymbol{\pi}^{\top}\mathcal{L}_\tau(\mathbf{q})=\mathbf{0}^{\top}$ for every $\mathbf{q}\in K$, we have
\[
\mathsf{P}_{\perp}\mathcal{L}_\tau(\mathbf{q})
=
\mathcal{L}_\tau(\mathbf{q})
=
\mathcal{L}_\tau(\mathbf{q})\mathsf{P}_{\perp}.
\]
Applying this fixed projection to \eqref{eq:closed-loop-outer-flow} therefore gives
\[
\frac{d}{ds}\mathbf{k}^{\perp}(s)
=
\boldsymbol{\phi}_\tau^{\perp}(s,\mathbf{k}(s))
-
\mathcal{L}_\tau(\mathbf{k}(s))\mathbf{k}^{\perp}(s).
\]
Reversibility makes $\mathcal{L}_\tau(\mathbf{q})$ self-adjoint and positive semidefinite in $L^2(\boldsymbol{\pi})$. Its restriction to the disagreement subspace satisfies the Poincar\'e inequality
\[
\left\langle
\mathbf{z}^{\perp},
\mathcal{L}_\tau(\mathbf{q})\mathbf{z}^{\perp}
\right\rangle_{L^2(\boldsymbol{\pi})}
\geq
\lambda_2(\mathbf{q})
\|\mathbf{z}^{\perp}\|_{L^2(\boldsymbol{\pi})}^2.
\]
Taking the $L^2(\boldsymbol{\pi})$ inner product of the projected flow with $\mathbf{k}^{\perp}(s)$, then applying Cauchy--Schwarz and the Poincar\'e inequality, yields \eqref{eq:closed-loop-dissipation}.

To obtain the integral estimate, set
\[
y(s):=\|\mathbf{k}^{\perp}(s)\|_{L^2(\boldsymbol{\pi})},
\qquad
b(s):=
\|\boldsymbol{\phi}_\tau^{\perp}(s,\mathbf{k}(s))\|_{L^2(\boldsymbol{\pi})}.
\]
Whenever $y(s)>0$, dividing \eqref{eq:closed-loop-dissipation} by $y(s)$ gives
\[
y'(s)
\leq
b(s)-\lambda_2(\mathbf{k}(s))y(s).
\]
At points where $y(s)=0$, the same scalar inequality holds in the upper-right Dini-derivative sense. The standard comparison principle therefore applies. Multiplying by the integrating factor
$\exp(\int_0^s\lambda_2(\mathbf{k}(u))\,du)$ and using $y(0)=0$ gives \eqref{eq:closed-loop-filter-bound}. Finally, Proposition~\ref{prop:spectral-gap-continuity} supplies the uniform lower bound $\underline{\lambda}_K>0$ on $K$, which yields the stated exponential envelope.}
\end{proof}

\rev{Theorem~\ref{thm:closed-loop-spectral-dissipation} gives the gap an endogenous meaning: $\lambda_2(\mathbf{k}(s))$ controls from below the instantaneous rate at which the played outer game dissipates disagreement, against injection by the selected running source. The trace term in \eqref{eq:selected-running-source} is the Riccati risk contribution; the entropy terms account for regularized macro actions. Holding the selected generator fixed recovers the usual frozen-generator convolution with $e^{-\lambda_2(s-r)}$. The result concerns the played gap along the equilibrium path; it does not assert that the game maximizes $\lambda_2$, that the gap is monotone in risk differences, or that entropy induces irreducibility, reversibility, or a common invariant distribution, which remain structural assumptions on the admissible rate matrices.}

\subsection{Relation to the Lyapunov--Metzler and Riccati--Metzler frameworks}\label{subsec:LM}

\rev{The inner LQ reduction \eqref{eq:coupled-Riccati} is, structurally, a coupled differential Riccati system for a Markov-jump linear system in which the off-diagonal transition rates enter through the Metzler operator $\mathcal{M}$. In its algebraic and synthesis forms this is closely related to the \emph{Lyapunov--Metzler} (LM) inequalities introduced by Geromel and Colaneri for switched linear and nonlinear systems \citep{GeromelColaneri2006,ColaneriGeromelAstolfi2008}, and to the \emph{Riccati--Metzler} (RM) conditions that arise when a continuous control input is added \citep{GeromelDeaectoDaafouz2013}. The two share the same Metzler-coupled quadratic structure; the correspondence and the point of departure are developed below.}

\rev{In the LM/RM framework one seeks symmetric $P_i\succ0$ and a Metzler matrix $\Pi=[\pi_{ij}]$ (row sums zero, $\pi_{ij}\ge0$ for $j\neq i$) satisfying a coupled inequality of the form $A_i^\top P_i + P_i A_i - P_i\Sigma^{ctrl}_iP_i + \sum_{j}\pi_{ij}P_j \prec 0$; the induced min-switching rule $\sigma=\arg\min_i x^\top P_i x$ then certifies stability or guaranteed performance. The equilibrium flow \eqref{eq:coupled-Riccati} is exactly of this Metzler-coupled Riccati type, with $\mathcal{M}(\mu^{*,\tau}(\mathbf{k}))$ in the role of the Metzler matrix. This reduction is classical: the LM/RM literature establishes that Markov-jump and switched LQ problems are naturally treated this way, and Proposition~\ref{prop:riccati-wellposed} is its finite-horizon, two-player counterpart.}

\rev{The defining feature of the present setting is that the Metzler coupling is \emph{not} selected by a single minimizing scheduler. In LM/RM, $\Pi$ is one decision variable optimized by one player (min-switching) to stabilize the plant. Here the rates are the equilibrium outcome of a \emph{two-player zero-sum game played on the generator}: under the bilinear parameterization $\mu_{ij}(f,g)=\bar\mu_{ij}+f^\top\Lambda_{ij}g$, the equilibrium operator $\mathcal{M}(\mu^{*,\tau}(\mathbf{k}))$ is generated row-wise by the unique entropy-regularized saddle in \eqref{eq:outer-regularized-hji}. This outer game is itself coupled to a robust inner differential game through the value-dependent running cost. As $\tau_D,\tau_A\downarrow0$, its value approaches that of the corresponding unregularized local matrix game; at positive temperature, entropy supplies a canonical continuous selection of the contested rates.}

\begin{proposition}[Reduction to the Riccati--Metzler and Lyapunov--Metzler equations]\label{prop:lm-rm-reduction}
\rev{Fix the outer policies at any constant admissible pair $(f^\circ,g^\circ)$, so the transition rates freeze at $\mu_{ij}^\circ=\mu_{ij}(f^\circ,g^\circ)$ and the Metzler operator $\mathcal{M}(\mu^\circ)$ is constant.
\begin{enumerate}
\item[(i)] \emph{(Riccati--Metzler.)} The coupled game Riccati flow \eqref{eq:coupled-Riccati} becomes
  \begin{equation*}
  \begin{aligned}
  -\dot P_i
  ={}& Q_i + A_i^\top P_i + P_i A_i
  - P_i\bigl(B_iR_i^{-1}B_i^\top - D_iS_i^{-1}D_i^\top\bigr)P_i \\
  &+ \sum_{j\neq i}\mu^\circ_{ij}\,(P_j-P_i),
  \qquad P_i(T)=Q_{T,i};
  \end{aligned}
  \end{equation*}
removing the inner disturbance channel ($D_i=0$, equivalently $S_i\to\infty$) leaves $-\dot P_i = Q_i + A_i^\top P_i + P_i A_i - P_i B_iR_i^{-1}B_i^\top P_i + \sum_{j\neq i}\mu^\circ_{ij}(P_j-P_i)$, the finite-horizon coupled Riccati--Metzler equation of the Markov-jump LQ regulator with generator $\mu^\circ$ \citep{GeromelDeaectoDaafouz2013}.
\item[(ii)] \emph{(Lyapunov--Metzler.)} Removing in addition the continuous control ($B_i=0$), a symmetric $P_i\succ0$ rendering the right-hand side negative definite satisfies the coupled Lyapunov--Metzler inequality $A_i^\top P_i + P_i A_i + \sum_{j}\pi_{ij}P_j\prec0$ with $\pi=\mu^\circ$, whose min-switching rule $\sigma=\arg\min_i x^\top P_i x$ is the certificate of \citep{GeromelColaneri2006,ColaneriGeromelAstolfi2008}.
\item[(iii)] \emph{(Strict extension.)} The present model is recovered from (i) by replacing the frozen $\mu^\circ$ with the rate path $\mu^{*,\tau}(\mathbf{k}(t))$ selected by the entropy-regularized outer saddle \eqref{eq:entropy-selected-rates}, and by reinstating the inner adversary $D_iS_i^{-1}D_i^\top\neq0$, which can make $\Sigma^{ctrl}_i$ indefinite. Neither modification is present in the LM/RM syntheses or in the one-designer dual-switching designs of Bolzern, Colaneri and De Nicolao \citep{BolzernColaneriDeNicolao2014,BolzernColaneriDeNicolao2016}, where deterministic scheduling and stochastic jumps coexist under a single designer; here the discrete transitions are themselves the contested object of a continuous-time minimax, co-determined with the inner robust controller.
\end{enumerate}}
\end{proposition}
\begin{proof}
\rev{Setting $(f,g)\equiv(f^\circ,g^\circ)$ in \eqref{eq:coupled-Riccati} freezes the rates and substitutes $\Sigma^{ctrl}_i=B_iR_i^{-1}B_i^\top-D_iS_i^{-1}D_i^\top$, giving the displayed flow; $D_i=0$ deletes the disturbance term and $B_i=0$ deletes the quadratic term, leaving the linear Lyapunov--Metzler operator, and the min-switching certificate follows from the coupled inequality in the standard way \citep{GeromelColaneri2006}. Part (iii) is immediate from \eqref{eq:entropy-selected-rates} and the definition of $\Sigma^{ctrl}_i$.}
\end{proof}

\rev{The comparison also clarifies the solution concept. In Lyapunov--Metzler theory the certificate is the piecewise-quadratic, generally non-smooth function $x\mapsto\min_i x^\top P_i x$. This is analogous to the non-smooth value envelopes that arise in optimal switching. The present framework instead keeps the regime-indexed viscosity objects $\{V_i\}$ and $\{U_i\}$ directly. In a degenerate one-designer instantaneous-switching limit one expects an envelope of the form $\min_i V_i$; establishing that singular limit is beyond the scope of this paper. The analogy positions the objects without identifying them.}

\rev{Finally, the comparison locates the new analytical content. In LM/RM stabilization the quadratic term carries only the control channel $B_iR_i^{-1}B_i^\top\succeq0$, and global solvability is an infinite-horizon linear-matrix-inequality (LMI) feasibility question. Reinstating the inner adversary replaces this term by $\Sigma^{ctrl}_i=B_iR_i^{-1}B_i^\top-D_iS_i^{-1}D_i^\top$, which can be indefinite; outside a disturbance-attenuation regime, the coupled flow \eqref{eq:coupled-Riccati} may exhibit finite escape and global existence on $[0,T]$ is no longer automatic. Proposition~\ref{prop:riccati-wellposed} supplies a sufficient global-solvability condition through an order-preserving comparison against the linear Lyapunov--Metzler flow $\Pi_i$. The Metzler coupling therefore enters as a time-varying operator co-determined with a potentially indefinite, finite-horizon Riccati flow.}

\rev{The positioning relative to the Lyapunov--Metzler and Riccati--Metzler literature can be stated along five axes. The Metzler rate matrix is a single design variable in LM/RM, whereas here it is the equilibrium of an entropy-regularized zero-sum game. The switching law is shaped by one minimizing scheduler in LM/RM, and by a competing attacker and stabilizer in our setting. The continuous layer is absent in the Lyapunov--Metzler case and is a single controller in the Riccati--Metzler case, whereas here it is a robust inner differential game. The coupling is single-level in LM/RM and bilevel in our formulation, with the inner value feeding the outer rate game. Finally, the LM/RM results are predominantly infinite-horizon LMI syntheses, whereas the present treatment is a finite-horizon coupled Riccati flow.}

\section{Application: Cross-layer Avellaneda-Stoikov Game }\label{sec:application}

To demonstrate the efficacy of the games-in-games architecture, we apply the framework to a high-frequency market making problem under regime uncertainty. We extend the classical Avellaneda–Stoikov (AS) inventory management model \citep{avellaneda2008high} to a hierarchical hybrid-systems framework with two coupled decision layers.
At the inner layer, a zero-sum differential game is played between a \emph{Market Maker} (MM), who controls inventory and quoting decisions, and a \emph{Strategic Predator} (SP), who acts adversarially by perturbing the short-term price drift.

At the outer layer, a separate strategic game governs the evolution of market regimes. The outer players, the \emph{macro-attacker} and the \emph{macro-stabilizer}, select discrete actions that parameterize the generator of a controlled Markov jump process over market regimes. To simplify the exposition, we assume binary actions for $\{ \texttt{off}, \texttt{stab}\}$ macro-stabilizer and $\{ \texttt{off}, \texttt{att}\}$ for macro-attacker. Hence the mixed strategies can be each parameterized by a single parameter within $[0,1]$. The macro strategy pair $(f_t, g_t)$ thus induces a transition-rate matrix governing switches between calm, volatile, and stressed market conditions. Through these rate controls, the outer game shapes the stochastic environment faced by the inner inventory game.

\rev{For concreteness, we record how the abstract primitives of Definition~\ref{def:hybrid} instantiate in this application. The continuous state is $X_t=(S_t,q_t,m_t)$ (mid-price, inventory, cash) and the regimes $\mathcal{I}$ are the market conditions (calm/volatile/stressed). The inner control--disturbance pair $(\nu,\omega)$, with realized values $(u,w)$, is the MM's quoted spread $u=(u^a,u^b)$ against the SP's drift perturbation $w$; the inner running and terminal costs specialize to the exponential-utility certainty equivalent of terminal wealth. The outer action sets $(\mathcal{A}_D,\mathcal{A}_A)$ are the stabilizer/attacker moves that set the transition-rate kernel $\mu_{ij}(f,g)$. Under this dictionary the inner Isaacs system \eqref{eq:HJI-inner} becomes the modified Avellaneda-Stoikov Hamilton-Jacobi-Bellman (HJB) equation, the outer system \eqref{eq:HJI-outer} becomes the macro Isaacs equation \eqref{eq:macroHJI}, and the exponential transform plays the role that the quadratic value ansatz plays in the LQ case of Section~\ref{sec:case}.}

\subsection{Market Dynamics \& Game Formulation}

Let the market operate in one of $N$ regimes, $I_t \in \mathcal{I}$. The mid-price $S_t \in \mathbb{R}_+$ follows a controlled diffusion process:
\begin{equation}
dS_t = w_t dt + \sigma(I_t) dW_t,
\end{equation}
where $\sigma(i)$ is the regime-dependent volatility, and $w_t$ is the drift controlled by the inner adversary SP.

The MM holds inventory $q_t \in \mathcal{Q} = \{-Q_{\max}, \dots, Q_{\max}\}$ and cash $m_t \in \mathbb{R}$. The inventory dynamics are pure jump processes driven by the execution of limit orders:
\begin{equation}
dq_t = dN^b_t - dN^a_t,
\end{equation}
where $N^b_t$ and $N^a_t$ are Poisson processes with intensities $\Lambda^b(u^b) = A e^{-k u^b}$ and $\Lambda^a(u^a) = A e^{-k u^a}$, controlled by the MM's spreads $u^b , u^a \in \mathbb{R}_{\geq 0}$ and the market depth parameters $A, k$.
Mapping to our general framework (Def.~\ref{def:hybrid}), we have the physical state: $X_t = (S_t, q_t, m_t)$, $x_0 = (S_0, q_0, m_0)$. Note that $S_t$ and $m_t$ are continuous, while $q$ is discrete. The MM is the \textit{micro-player} who chooses the spread $u_t = ( u^a_t, u^b_t)$; the predator is the \textit{micro-adversary} who chooses $w_t$ (price drift). The jump rate is state-independent, with magnitude $\rho (z)=1$. 

At the outer layer, the regime transition rates $\mu_{ij}(f,g)$ are controlled by a \textit{Macro-Attacker} ($f_t$) who seeks to maximize the MM's disutility (or induce a ``Crisis'' regime where $\sigma$ is high); and a \textit{Stabilizer} ($g_t$) who seeks to maintain the ``Calm'' regime.
The outer value function $U_i(t,q)$ is computed by substituting the inner value $v_{i,q}$ into the outer objective.

We formulate the inner layer as a zero-sum differential game. The MM maximizes the Constant Absolute Risk Aversion (CARA) utility of terminal wealth, $U(x_0) = -\exp(-\gamma (m_T + q_T S_T))$, where $\gamma \geq 0$ is the risk aversion parameter.
The SP observes the MM's inventory $q_t$ and exerts price pressure $w_t$ to minimize the MM's Certainty Equivalent, \rev{with capital usage or manipulation risk represented below by a utility-scaled quadratic penalty with coefficient $\frac{1}{2\xi}$.}

\subsection{Hierarchical Solution: Matrix Exponential and Approximate Equilibrium}

Using the Ansatz $V_i(t, S, q, m) = -\exp(-\gamma(m + qS + \theta_i(t,q)))$, \rev{where $\theta_i(t,q)$ denotes the \emph{inventory risk function} in regime $i$, namely the certainty-equivalent reservation adjustment (in price units) that the MM applies for carrying inventory $q$ to the horizon,} the inner Hamiltonian $\mathcal{H}_i$ decomposes additively due to the separation of drift (price) and jump (execution) controls.
The SP chooses the price-drift perturbation $w$ to minimize the MM's value, trading the marginal price impact against a quadratic manipulation cost $\tfrac{1}{2\xi}w^2$. \rev{Because the value $V_i<0$ is an exponential (CARA) utility, the running manipulation cost is measured in utils, i.e.\ it enters scaled by the local marginal disutility $-V_i>0$; the relevant Hamiltonian component associated with the price drift is therefore}
\begin{equation*}
\min_{w} \Bigl[\, w\,\partial_S V_i \;-\; V_i\,\tfrac{1}{2\xi}\,w^2 \,\Bigr].
\end{equation*}
\rev{Substituting $\partial_S V_i = -\gamma q\, V_i$ and dividing through by $-V_i>0$ reduces this to the scalar convex problem $\min_{w}\{\gamma q\,w + \tfrac{1}{2\xi}w^2\}$, whose first-order condition $\gamma q + \xi^{-1}w = 0$ yields the closed-form structural reaction function}
\begin{equation}
\label{eq:predator-strat}
w^*(t, q) = -\xi \gamma q.
\end{equation}
This strategy reveals a \emph{mean-averting} behavior: if the MM is long ($q>0$), the SP pushes the price down ($w^* < 0$) to devalue the position; if short, the SP pushes the price up.

Simultaneously, the MM maximizes the trading component:
\begin{equation}
\max_{u^a, u^b \in \mathbb{R}_{\geq 0}} \sum_{side \in \{a,b\}} \Lambda^{side}(u^{side}) \left( 1 - e^{-\gamma(u^{side} + \Delta \theta_{side})} \right),
\end{equation}
where $\Delta \theta_{a} = \theta(q-1)-\theta(q)$ and $\Delta \theta_{b} = \theta(q+1)-\theta(q)$. This recovers the standard AS spread formula adjusted for inventory shadow cost.
Substituting the optimal strategies back into the HJB equation, the predatory term contributes a quadratic penalty scaled by inventory size:
\[
w^* (-\gamma q) - \frac{1}{2\xi}(w^*)^2 = \frac{1}{2}\xi \gamma^2 q^2.
\]
Consequently, the inventory value function $\theta_i(t,q)$ satisfies the coupled system of ODEs:
\begin{align}
\label{eq:modified-AS-HJB}
-\dot{\theta}_i(t,q) &= \underbrace{\frac{1}{2}\gamma \sigma_i^2 q^2}_{\text{Volatility Risk}} + \underbrace{\frac{1}{2}\xi \gamma^2 q^2}_{\text{Predatory Risk}}  + \sum_{side \in \{a,b\}} \frac{A}{\gamma} \left( 1 + \frac{\gamma}{k} \right)^{-\frac{k}{\gamma}} e^{-\gamma \Delta \theta_{side}} \nonumber \\
&+ \sum_{j \neq i} \mu_{ij}(f,g) \frac{1}{\gamma} \left( 1 - e^{-\gamma(\theta_j - \theta_i)} \right).
\end{align}

Under CARA utility, we define the value function and its exponential transformation as:
\begin{equation*}
V_i(t,S,q,m) = -\exp\!\big(-\gamma[m+qS+\theta_i(t,q)]\big), \qquad
v_{i,q}(t) := e^{-\gamma\theta_i(t,q)}.
\end{equation*}
Substituting optimal quotes into the HJB equation reduces the system to a linear ODE:
\begin{equation*}
\dot{v}(t) = M(t)v(t), \qquad
v(T) = \mathbf{1}.
\end{equation*}
The generator $M$ acts on the state space stacked by regimes $i=1,\dots,N$ and inventory $q \in \{-Q_{\max},\dots,Q_{\max}\}$. It decomposes into micro-structure and macro-switching blocks:
$M = D + \big(Q \otimes I_{|\mathcal{Q}|}\big).$
Here, $Q$ is the regime transition matrix ($Q_{ij}=\mu_{ij}$ for $i\neq j$, row-sums zero). The block-diagonal matrix $D = \mathrm{diag}(A_1,\dots,A_N)$ captures the micro-dynamics. Each block $A_i$ is tridiagonal in the inventory dimension $q$, containing
the diagonal risk \& outflow: $\frac{1}{2}\gamma^2(\sigma_i^2+\xi\gamma)q^2 - (\Lambda^{a}_i + \Lambda^{b}_i)$, and off-diagonal entries $\Lambda^{\text{side}}_i e^{-\gamma u^{*,\text{side}}_i}$ at $(q, q\pm 1)$.

For outer parameters piecewise-constant on $[t,T]$, the solution is explicit:
\begin{equation}
v(t) = \exp (- M\tau )\mathbf{1}, \qquad \tau := T-t.
\end{equation}

\paragraph{Regime mixing and expected variance}
\rev{For small horizons $\tau$, we use the Feynman--Kac representation to isolate how expected accumulated variance affects pricing. Retaining that variance channel explicitly and the leading monopoly-rent term gives the approximation} (when $q$ is at the boundaries, remove the coefficient $2$ in the \emph{monopoly rent} term):
\begin{equation}
\label{eq:theta-first}
\theta_i(t,q)
= \frac{q^2}{2}\Big(\gamma\,w_i(\tau)+\gamma^2\xi\,\tau\Big)
- \frac{2A}{\gamma}\Big(1+\frac{\gamma}{k}\Big)^{-k/\gamma}\tau
+ \mathcal{O}(\tau^2).
\end{equation}
Here, $w_i(\tau)$ represents the \emph{expected integrated variance} starting from regime $i$. By expanding the generator $e^{Qu} \approx I + Qu$, we explicitly capture the regime mixing effect. Letting $s_i := \sigma_i^2$:
\begin{align}
\label{eq:wtau-expansion}
w_i(\tau) &:= \int_0^\tau [e^{Qu} s]_i \, du \approx \int_0^\tau [(I + Qu) s]_i \, du \nonumber \\
&= \sigma_i^2\tau + \frac{1}{2}\sum_{j \neq i} \mu_{ij}(\sigma_j^2 - \sigma_i^2)\,\tau^2 + \mathcal{O}(\tau^3).
\end{align}
\rev{Within the expected-variance channel, the $\mathcal{O}(\tau^2)$ term in \eqref{eq:wtau-expansion} captures the \emph{Regime Risk}: the probability-weighted drift into different volatility states. The remainder in \eqref{eq:theta-first} also collects other higher-order microstructure interactions, so \eqref{eq:theta-first} is not a complete second-order expansion of the quoting system.} Thus, for very short horizons, the MM prices using the current regime's volatility. As $\tau$ increases, the pricing formula ``bends'' to incorporate the volatilities of connected regimes.

\begin{remark}[Risk Isomorphism]
The Hamiltonian separability preserves the tractability of the solution while providing a key insight. The presence of a Strategic Predator ($\xi > 0$) is mathematically isomorphic to an increase in market volatility. The MM perceives an \emph{effective volatility} \rev{$\sigma_{eff}^2(i) = w_i(\tau)/\tau + \xi \gamma$, equivalently $C_i(\tau)=\gamma\tau\,\sigma_{eff}^2(i)$ with $C_i(\tau)$ defined below, which reduces to $\sigma_i^2 + \xi\gamma$ at leading order in $\tau$.} This implies that in the presence of predatory order flow, the optimal policy is to widen spreads and liquidate inventory faster, exactly as one would in a high-volatility environment.
\rev{Three implications follow for strategy and parameter interpretation. First, the predator cost $\xi$ and the risk aversion $\gamma$ enter only through the single scalar $\xi\gamma$, so adversarial order flow is observationally equivalent to a more risk-averse MM facing a more volatile tape: without an instrument, calibration cannot separate ``true'' volatility from predation, and a practitioner may fold an estimate of predatory intensity directly into the volatility input of the standard AS formula. Second, because $\sigma_{\mathrm{eff}}^2$ is additive in $\xi\gamma$, the optimal half-spread retains the AS form with $\sigma_i^2$ replaced by $\sigma_{\mathrm{eff}}^2(i)$, so no separate predator-tracking machinery is needed at the quoting layer. Third, the isomorphism is regime-local: since $w_i(\tau)$ already blends the connected regimes' volatilities (cf.\ \eqref{eq:wtau-expansion}), the regime-blended variance-channel approximation predicts pre-emptive widening in regimes that are one transition away from a high-$\sigma$ state. The leading-order counterfactual in Section~\ref{sec:application} isolates the instantaneous effective-volatility channel and does not measure this forward-looking effect.}
\end{remark}

\paragraph{Optimal quotes and time-varying spreads}
Defining the effective risk factor $C_i(\tau) := \gamma w_i(\tau) + \gamma^2\xi \tau$, the inventory indifference pricing implies:
$\Delta\theta_a \approx \Big(\tfrac{1}{2}-q\Big)C_i(\tau), \qquad
\Delta\theta_b \approx \Big(q+\tfrac{1}{2}\Big)C_i(\tau).$
The resulting optimal total spread $u^*$ is:
\begin{equation*}
u^*_{i}(t,q) = u^{*, a}_i + u^{*, b}_i= \frac{1}{\gamma}\ln\!\Big(1+\frac{\gamma}{k}\Big) + \frac{1}{2}C_i(\tau), 
\end{equation*}
\rev{Within the variance-channel approximation, this demonstrates that spreads widen pre-emptively based on future expected volatility $w_i(\tau)$, pricing in the macro-attacker's threat before the regime shift occurs.}

\subsection{Macro Equilibrium (Outer HJI)}

\rev{To retain the notation customary in this application, let $f_t$ denote the Attacker and $g_t$ the Stabilizer; this reverses the local role labels used for $(f_i,g_i)$ in the LQ specialization above, and the minimax order below is correspondingly $\min_g\max_f$.} These controls shape the regime generator $\mu_{ij}(f,g)$ and micro-parameters.
We posit that the macro-agents optimize against the \emph{anticipated} inventory cost priced in by the Market Maker. Thus, the outer-layer running cost $\varphi_i(q;f,g)$ is defined directly by the risk function $\theta_i$ over the horizon $\tau$:
\begin{equation*}
\varphi_i(q;f,g) := \theta_i(t,q) \approx \frac{q^2}{2}\Big(\gamma w_i(\tau; f,g)+\gamma^2\xi\tau\Big) - \frac{2A}{\gamma}\Big(1+\frac{\gamma}{k}\Big)^{-k/\gamma}\tau.
\end{equation*}
This modeling choice implies a \emph{sentiment-driven interaction}: the macro-attacker targets the Market Maker's forward-looking risk assessment, the channel driving liquidity drying, with instantaneous volatility serving only as one input.

The macro value function $U_i(t,q)$, representing the cumulative market stress, satisfies the Isaacs equation:
\begin{equation}
\label{eq:macroHJI}
\begin{aligned}
-\partial_t U_i(t,q) = \min_{g\in\Delta(\mathcal{A}_D)} \max_{f\in\Delta(\mathcal{A}_A)} \Big\{ &\varphi_i(q;f,g) \\
&{}+ \sum_{j\neq i}\mu_{ij}(f,g) \big(U_j(t,q) - U_i(t,q)\big) \Big\}.
\end{aligned}
\end{equation}
\begin{remark}[Behavioral Interpretation]
By utilizing the integrated variance $w_i(\tau)$ within the running cost, this formulation creates a feedback loop where the macro-agents are highly sensitive to \emph{future} regime risks. The switching probability enters twice: once in the MM's pricing ($w_i$) and again in the outer value dynamics ($\sum \mu_{ij} \Delta U$). This creates a ``Hyper-Alert'' equilibrium where attackers preemptively strike as soon as the \emph{expectation} of future volatility rises, mirroring the self-fulfilling nature of liquidity crises.
\end{remark}

Let $\Delta_{ij}(t,q) := U_j(t,q) - U_i(t,q)$ be the stability gap (the cost impact of switching from $i$ to $j$).

\begin{enumerate}
    
\item \emph{Affine Control:} If the transition rates $\mu_{ij}$ are controllable within $[\underline{\mu}_{ij}, \overline{\mu}_{ij}]$ via $\mu_{i j}(f, g)=\mu_{i j}^{0}+f \cdot \lambda_{i j}^{\mathrm{att}}-g \cdot \lambda_{i j}^{\mathrm{stab}}$ \rev{(where $\mu_{ij}^{0}\ge 0$ is the nominal baseline rate from $i$ to $j$, and $\lambda_{ij}^{\mathrm{att}},\lambda_{ij}^{\mathrm{stab}}\ge 0$ are the per-unit-effort rate sensitivities of the attacker and stabilizer, respectively)}, the optimization decouples into pointwise bang-bang switches.
The Attacker ($f$) maximizes the drift toward higher cost regimes, while the Stabilizer ($g$) minimizes it:
\begin{align*}
f^*(t) = \rev{\mathbb{I}_{\left\{ \sum_{j \neq i} \lambda_{ij}^{\text{att}} \Delta_{ij} > 0 \right\}}},  \qquad
g^*(t)  = \rev{\mathbb{I}_{\left\{ \sum_{j \neq i} \lambda_{ij}^{\text{stab}} \Delta_{ij} > 0 \right\}}}.
\end{align*}
This highlights the conflict: when a regime switch is dangerous ($\Delta_{ij} > 0$), the Attacker pushes the accelerator ($\overline{\mu}$) while the Stabilizer slams the brake ($\underline{\mu}$).

\item \emph{Quadratic Costs:} We relax the bounded control assumption and instead impose quadratic effort penalties $\frac{\rho_f}{2}f^2$ and $\frac{\rho_g}{2}g^2$ in the outer Hamiltonian. Assuming the transition rates remain affine in effort ($\mu_{ij} = \mu_{ij}^0 + f\lambda_{ij}^{\text{att}} - g\lambda_{ij}^{\text{stab}}$), the control-dependent part of the Hamiltonian is:
\begin{equation*}
\mathcal{H}(f,g) \propto f \left( \sum_{j \neq i} \lambda_{ij}^{\text{att}} \Delta_{ij} \right) - g \left( \sum_{j \neq i} \lambda_{ij}^{\text{stab}} \Delta_{ij} \right) - \frac{\rho_f}{2}f^2 \rev{+ \frac{\rho_g}{2}g^2}.
\end{equation*}
The first-order optimality conditions ($\partial_f \mathcal{H} = 0, \partial_g \mathcal{H} = 0$) \rev{use the sign convention that the maximizer's effort penalty enters with a $-$ sign and the minimizer's with a $+$ sign, making $\mathcal{H}$ concave in $f$ and convex in $g$. They} yield explicit proportional feedback rules:
\begin{align*}
f^*(t) &= \frac{1}{\rho_f} \left[ \sum_{j \neq i} \lambda_{ij}^{\text{att}} \big( U_j(t,q) - U_i(t,q) \big) \right]^+, \\
g^*(t) &= \frac{1}{\rho_g} \left[ \sum_{j \neq i} \lambda_{ij}^{\text{stab}} \big( \rev{U_j(t,q) - U_i(t,q)} \big) \right]^+,
\end{align*}
where $[x]^+ = \max(0, x)$.
This result characterizes the macro-agents as \emph{variable-gain controllers}: the intensity of their intervention scales linearly with the severity of the stability gap. For instance, the Attacker exerts minimal effort when the system is robust ($\Delta_{ij} \approx 0$) but surges activity proportionally as the MM's inventory vulnerability increases ($\Delta_{ij} \gg 0$).
\end{enumerate}

\subsection{Numerical Illustration}

We demonstrate the equilibrium strategies and Risk Isomorphism principle using calibrated Bitcoin (BTC) market data. The experiment compares two market making strategies facing a Strategic Predatory trader: 1.\textit{vanilla AS}: standard AS strategy, unaware of predatory drift; 2.\textit{equilibrium AS}: modified strategy using the leading-order effective volatility \rev{$\sigma_{\text{eff},i}^2 = \sigma_i^2 + \xi \gamma$} to account for predatory risk.

\begin{figure}
    \centering
    \includegraphics[width=\linewidth]{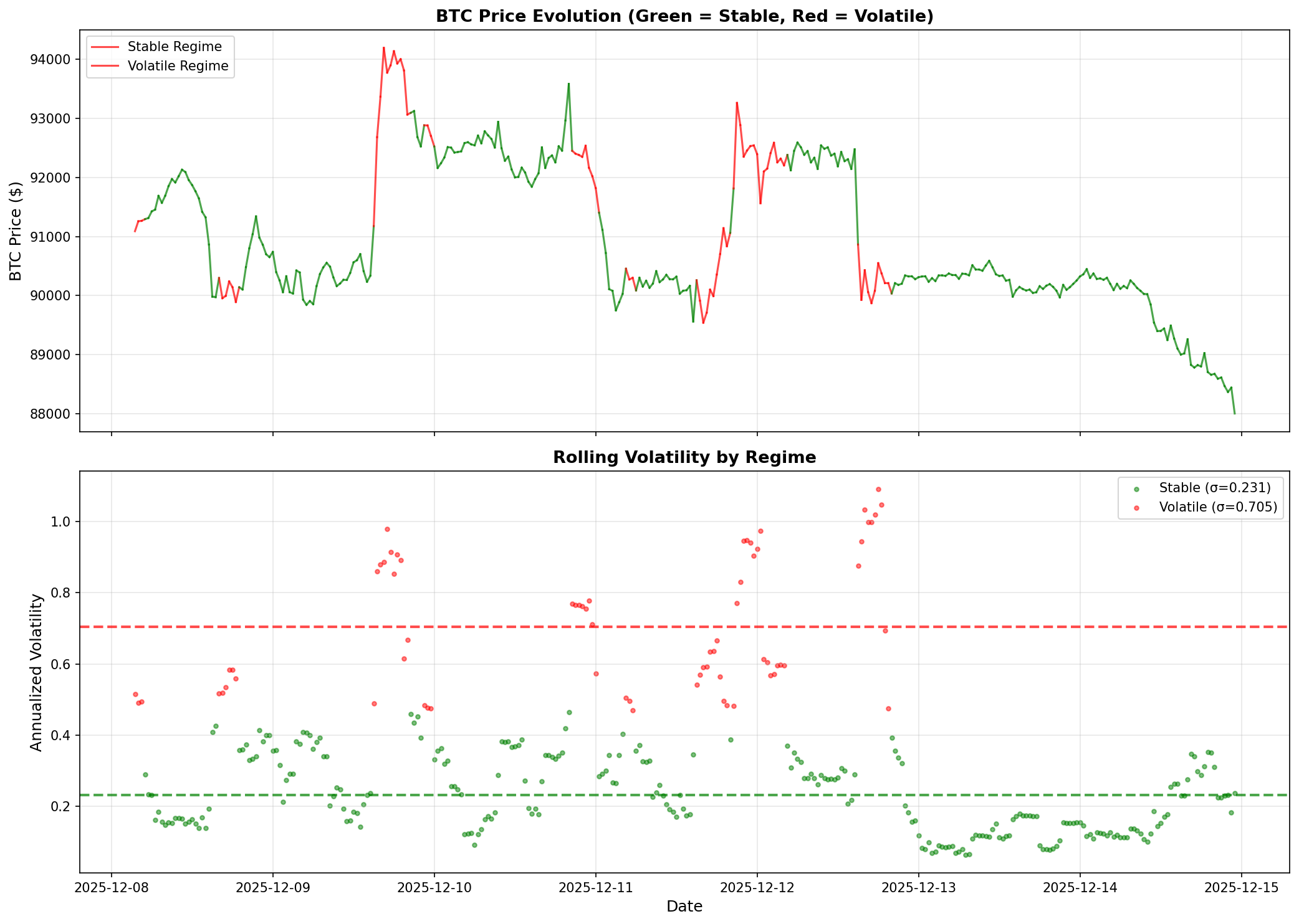}
    \caption{\rev{Regime calibration from Kraken BTC--USD $30$-minute OHLCV data (December 7--14, 2025; $328$ observations). Rolling realized volatility is clustered ($K$-means, $K=2$) into a stable regime ($\sigma_0=23.14\%$ annualized, green; $78.7\%$ of the sample) and a volatile regime ($\sigma_1=70.50\%$ annualized, red; $21.3\%$). Top: price path colored by the inferred regime; bottom: rolling volatility with the two regime means. The empirical transition rates are $\lambda_{01}=1.49$ and $\lambda_{10}=6.17~\mathrm{day}^{-1}$ (symmetric mean $\mu_0=3.83~\mathrm{day}^{-1}$).}}
    \label{fig:btccalibration}
\end{figure}

We calibrate regime-switching parameters from Kraken ticker ``BTC-USD'' 30-minute OHLCV data (December 7--14, 2025). Using rolling volatility with $K$-means clustering, we identify two distinct regimes:
1. \textit{stable regime} (regime 0): $\sigma_0 = 0.2314$ (23.14\% annualized); 2.\ \textit{volatile regime} (regime 1): $\sigma_1 = 0.7050$ (70.50\% annualized). The volatility ratio is $\sigma_1 / \sigma_0 = 3.05$.

\rev{Empirical transition-matrix estimation yields asymmetric rates $\lambda_{01}=1.49$ and $\lambda_{10}=6.17$ per day (stable-to-volatile and volatile-to-stable), i.e.\ a mean base rate $\mu_0=3.83$ per day, a symmetric-mean holding-time scale of about $6.3$ hours, and a shorter-lived volatile regime ($21.3\%$ occupancy). The robustness study (Table~\ref{tab:robustness}) uses these calibrated \emph{asymmetric} rates directly, so its stationary volatile occupancy is $\lambda_{01}/(\lambda_{01}+\lambda_{10})\approx19\%$, close to the observed sample share; the single illustrative trajectory of the next subsection uses the symmetric mean rate $\mu_0$ as a modelling simplification.} Figure~\ref{fig:btccalibration} shows the calibrated regime evolution with price dynamics colored by regime state (green for stable, red for volatile).

\subsubsection{Counterfactual Simulation Design}

We simulate 12 hours of market making activity (December 12, 2025, 15:00--03:00) with the following setup: starting in stable regime ($I_0 = 0$) with initial price $S_0 = \$90{,}863.90$,  
  the market making process lasts for 2,880 total steps, with each step counting for $\Delta t = 15$ seconds.
  We set the price drift to be $0$ so that only the predator affects the drift, whose optimal drift control is $w^*(q) = -\xi \gamma q$ with cost coefficient $\xi = 10.0$.
MM's risk aversion parameter and inventory constraint are set to be $\gamma = 0.02$, $q \in [-10, 10]$. 
  We fit the order arrival to follow Poisson intensity $\Lambda(u) = \lambda_0 e^{-k  u }$ with market depth $A = 250{,}000$ per year, spread sensitivity $k = 10$. 
We simulate 1,000 Monte-Carlo paths for statistical significance.

Both strategies face the same Strategic Predator who observes their inventory in real-time and applies adversarial drift. The key difference is that Vanilla AS uses the actual volatility $\sigma_i$ in the spread formula, while Equilibrium AS uses the leading-order effective volatility \rev{$\sigma_{\text{eff},i} = \sqrt{\sigma_i^2 + \xi \gamma}$} derived from Remark~1 (Risk Isomorphism).

\begin{figure}[htbp]
    \centering
    \includegraphics[width=\linewidth]{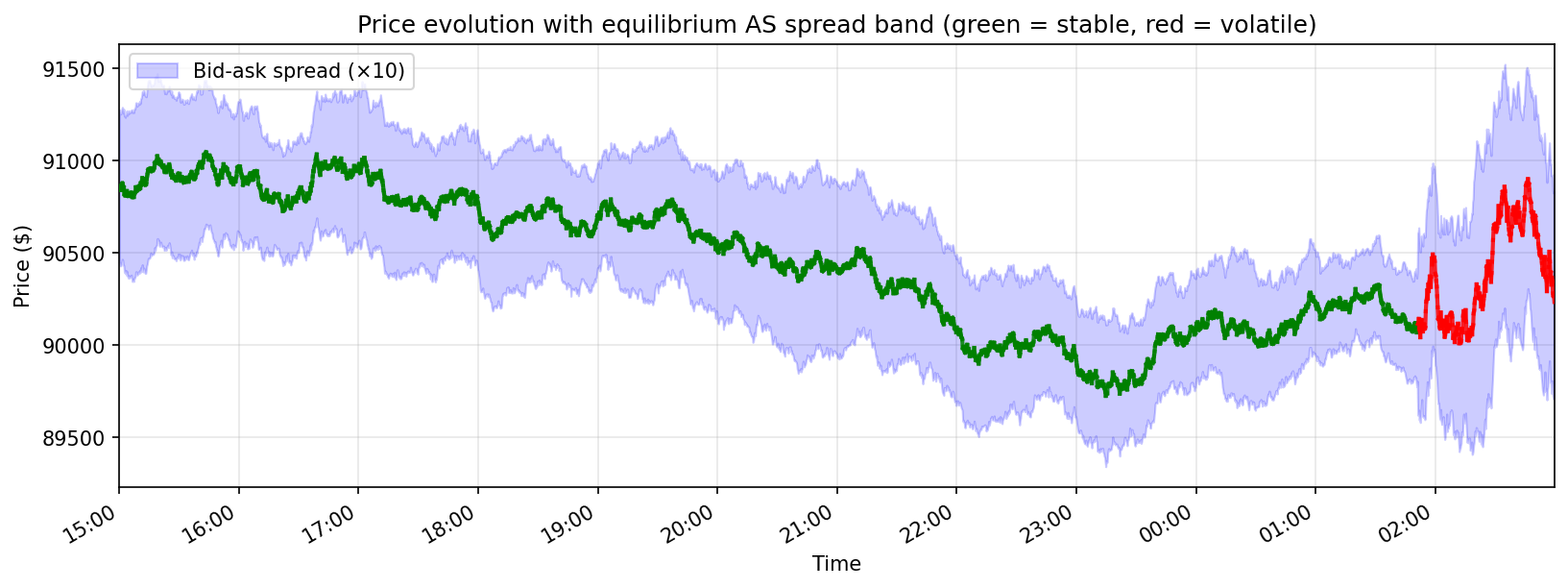}
    \caption{\rev{Illustrative single trajectory:} sample BTC price evolution with equilibrium AS spread bands (10 times the actual spread for clearer visualization).}
    \label{fig:mm-simulation}
\end{figure}

\subsubsection{Results and Behavioral Analysis}

Figure~\ref{fig:mm-simulation} presents \rev{one illustrative} price-evolution path for the counterfactual simulation, with equilibrium AS spread bands, colored by regime state (green for stable, red for volatile). \rev{It alternates between the calibrated stable ($\sigma_0=23.14\%$) and volatile ($\sigma_1=70.50\%$) regimes of Figure~\ref{fig:btccalibration}. This figure and the two that follow are \emph{single-trajectory illustrations} of the mechanism; our primary, statistically-controlled performance evidence is the multi-condition robustness study below (Table~\ref{tab:robustness}, Figure~\ref{fig:robustness}).}

\begin{figure}[htbp]
    \centering
    \includegraphics[width=\linewidth]{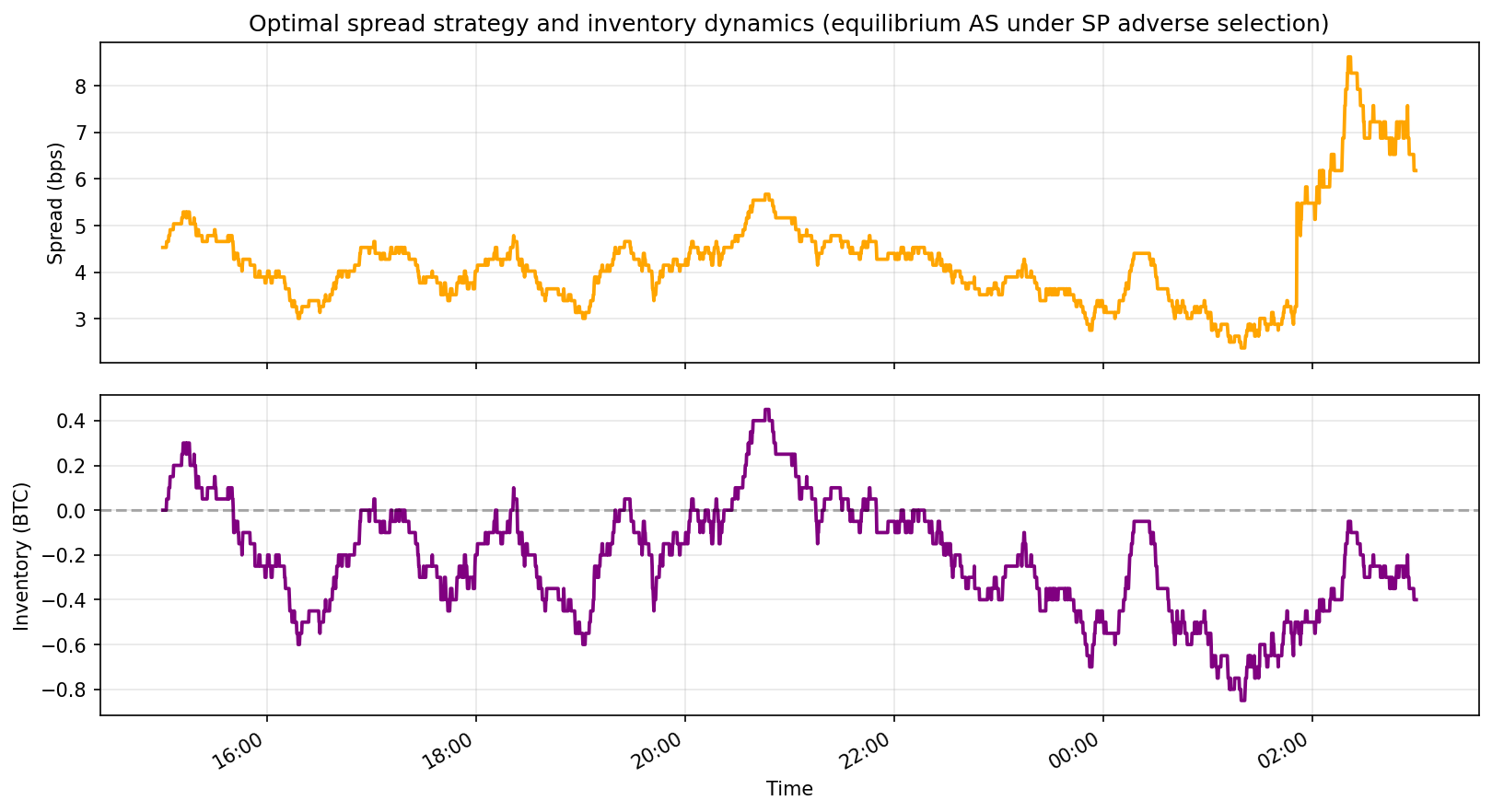}
    \caption{\rev{Illustrative single trajectory:} optimal spread strategy evolution and inventory dynamics of MM under SP adverse selection.}
    \label{fig:inventorydyn}
\end{figure}

Figure~\ref{fig:inventorydyn} shows the dynamic spread adjustment and inventory management. The equilibrium strategy adaptively widens spreads in volatile regimes, incorporating both the heightened market volatility $\sigma_i$ and the predatory risk $\xi \gamma$ into the reservation price calculation.
\rev{Section~\ref{subsec:spectral-analysis} identifies a separate outer-layer prediction: when the macro rate game is played, the selected generator is an endogenous function of continuation-value gaps. In the reversible LQ setting, its spectral gap varies continuously with those gaps and controls the instantaneous dissipation of regime-value disagreement along the closed-loop path. The present counterfactual does not measure this effect: it fixes the macro actions $f=g=0$, so it neither plays an outer rate game nor computes an equilibrium generator or its spectral gap. The spread widening reported here is driven by the inner effective-volatility channel (Risk Isomorphism); extending the AS experiment to a played macro game and tracking the induced generator is left to future work.}

\rev{Figure~\ref{fig:pnl} shows the terminal profit-and-loss (PnL) distribution across the Monte-Carlo paths for this single calibrated window; consistent with the mechanism, the equilibrium strategy shifts the distribution to the right of vanilla AS. This calibrated-window run is illustrative; the controlled quantification across process-parameter regimes, with confidence intervals (CIs), is reported in the robustness study below (Table~\ref{tab:robustness}, Figure~\ref{fig:robustness}).}

\begin{figure}[htbp]
    \centering
    \includegraphics[width=\linewidth]{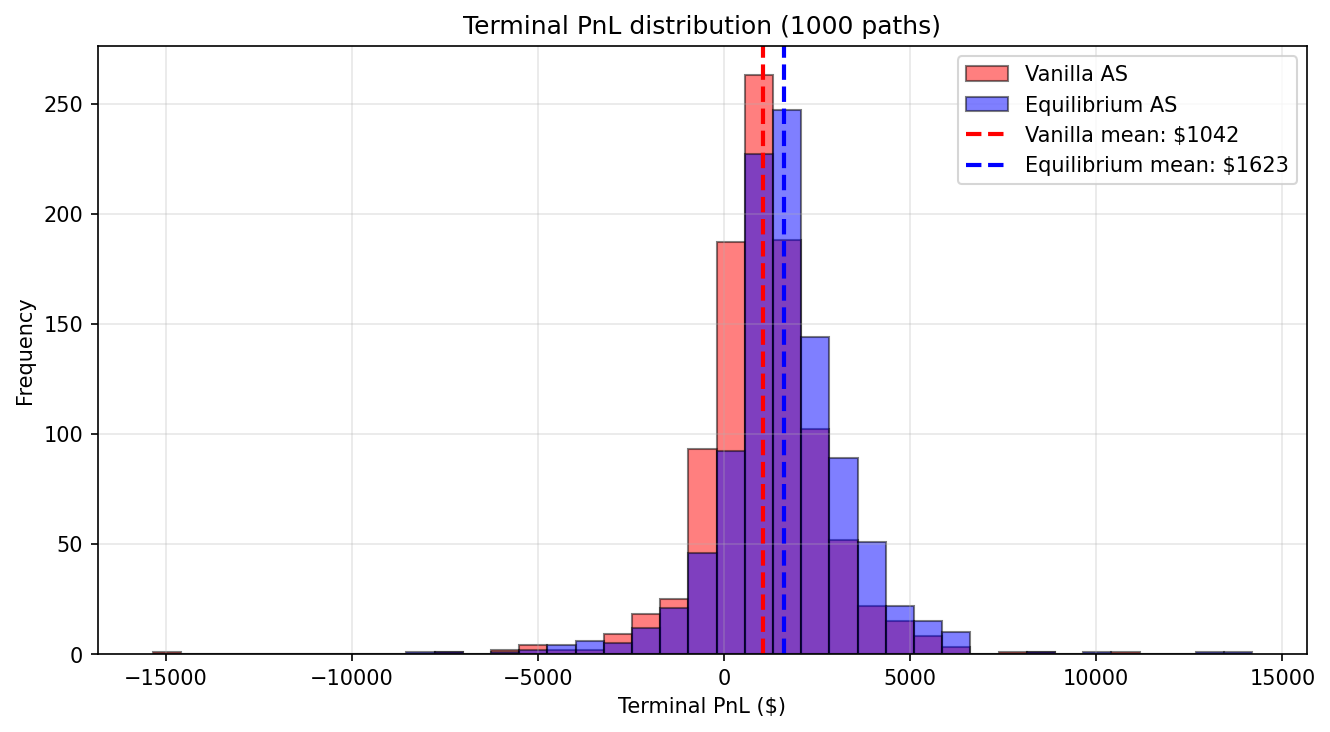}
    \caption{\rev{Illustrative single calibrated-window} distribution of terminal PnL over 1,000 simulated 12-hour episodes (vanilla AS vs equilibrium AS with predator); aggregate multi-condition results are in Table~\ref{tab:robustness}.}
    \label{fig:pnl}
\end{figure}

\rev{The mechanism is the expected one: equilibrium AS quotes wider spreads by pricing the predator into the effective volatility $\sigma_{\mathrm{eff}}^2=\sigma_i^2+\xi\gamma$, and accepts somewhat larger inventory, trading a modest reduction in fill rate (via $\Lambda(u)=\lambda_0 e^{-ku}$) for higher per-fill revenue. The robustness study below reports the net effect across paths, with confidence intervals; the single trajectory is used only as a mechanism illustration.}

\subsubsection{\rev{Robustness across market-condition regimes}}
\rev{Because the available Kraken feed spans only a seven-day window, genuinely independent calm, volatile, and stressed \emph{historical} windows cannot be collected. We therefore probe robustness through a process-parameter sweep: the regime-switching jump-diffusion is calibrated \emph{once} from the seven-day data (Figure~\ref{fig:btccalibration}), and we then sweep the \emph{process parameters} that define a market condition, namely the volatility ratio $\sigma_{\mathrm{vol}}/\sigma_{\mathrm{stab}}$ (regime severity), the base switching rate $\mu_0$ (regime persistence), and the predator strength $\xi\gamma$ (adversarial intensity). For each parameter regime we run an $800$-path Monte-Carlo counterfactual of vanilla versus equilibrium AS facing the same Strategic Predator, using common random numbers across the two strategies for variance reduction (so that the $95\%$ confidence intervals below are computed from paired path differences). The counterfactual uses the \emph{leading-order}, spread-only form of the equilibrium quote: the instantaneous effective half-spread with $\sigma_{\mathrm{eff}}^2(i)=\sigma_i^2+\xi\gamma$ and symmetric bid/ask. It does not use the full regime-blended ($w_i(\tau)/\tau$), inventory-asymmetric policy derived in Sections~\ref{sec:case}--\ref{sec:application}; both strategies are quoted on this same footing, so the comparison is like-for-like.}

\rev{Table~\ref{tab:robustness} reports the outcome. All regimes are calibrated once from the canonical seven-day parameters of Figure~\ref{fig:btccalibration}, using the \emph{asymmetric} transition rates $\lambda_{01}=1.49$, $\lambda_{10}=6.17~\mathrm{day}^{-1}$ directly; the persistence axis scales both rates by a common factor, holding the stationary volatile occupancy at the calibrated $\lambda_{01}/(\lambda_{01}+\lambda_{10})\approx19\%$ (close to the $\approx21\%$ observed sample share) throughout. The equilibrium strategy dominates the vanilla strategy in \emph{every} regime: the mean-PnL gain ranges from $+55\%$ to $+86\%$, all four (normalized) $95\%$ paired-difference intervals exclude zero, and equilibrium AS earns the higher terminal PnL on $83$--$88\%$ of paired paths; its Sharpe-like ratio (terminal-PnL mean/std) exceeds vanilla's in every regime ($\approx\!0.84$--$1.10$ versus $\approx\!0.58$--$0.78$). The one-axis sweeps in Figure~\ref{fig:robustness}(b,c) isolate the drivers. The gain increases monotonically with the predator strength $\xi\gamma$: as $\xi\gamma$ increases from $0.1$ to $0.8$, the mean-PnL improvement rises through $+44\%$, $+86\%$, $+161\%$, and $+331\%$, and the Sharpe-like improvement from $+28\%$ to $+138\%$. This is the quantitative signature of the Risk Isomorphism principle, since the value of pricing in predation grows with its intensity. The gain decreases as the volatility ratio grows, falling from $+98\%$ to $+48\%$ as $\sigma_{\mathrm{vol}}/\sigma_{\mathrm{stab}}$ increases from $1.5$ to $6.0$, because a large $\sigma_i^2$ dilutes the fixed additive correction in $\sigma_{\mathrm{eff}}^2=\sigma_i^2+\xi\gamma$; in the separate persistence sweep it varies moderately with switching speed, remaining positive at $+49\%$--$+64\%$ as $\bar\mu$ ranges over $2$--$16~\mathrm{day}^{-1}$ at fixed occupancy. The improvement is therefore positive with paired-difference intervals excluding zero across all tested process-parameter regimes, supporting the robustness of the CARA mechanism beyond the baseline calibration. We treat this controlled multi-condition study, in preference to any single-window run, as our primary numerical evidence.}

\rev{The persistence axis also supplies a limited empirical bridge to the spectral analysis of Section~\ref{subsec:spectral-analysis}. For the exogenous two-state chain used in this counterfactual, the nonzero eigenvalue of the negative generator is
\[
\lambda_2^{\mathrm{exo}}
=
\lambda_{01}+\lambda_{10}
=
2\bar{\mu}.
\]
At the calibrated rates, $\lambda_2^{\mathrm{exo}}=7.66~\mathrm{day}^{-1}$, corresponding to a relaxation timescale of approximately $3.1$ hours. Scaling both rates proportionally varies this baseline gap from approximately $4$ to $32~\mathrm{day}^{-1}$ while preserving stationary occupancy; the equilibrium-AS gain remains positive throughout. This is an exogenous spectral-gap stress test showing that the leading-order Risk-Isomorphism comparison is robust across regime-mixing timescales. It is not a numerical verification of Theorem~\ref{thm:closed-loop-spectral-dissipation}, which concerns the endogenous generator selected when the macro rate game is played.}

\begin{table}[htbp]
\centering
\caption{\rev{Robustness of the equilibrium strategy across process-parameter regimes (calibrated once from the canonical seven-day data; $800$ Monte-Carlo paths, paired common random numbers, fixed seed; the transition chain uses the calibrated \emph{asymmetric} rates $\lambda_{01}=1.49,\lambda_{10}=6.17~\mathrm{day}^{-1}$, and the persistence column $\bar\mu$ scales both while holding the stationary volatile occupancy at $\approx19\%$). Equilibrium AS dominates vanilla AS in every regime. The reported intervals are \emph{normalized} paired-difference $95\%$ CIs (the paired absolute-PnL CI divided by the vanilla mean; not an exact ratio CI); ``Eq.\ wins'' is the fraction of paired paths on which equilibrium AS has the higher terminal PnL; and ``Sharpe-like'' is the terminal-PnL mean-to-standard-deviation ratio (not an annualized Sharpe).}}
\label{tab:robustness}
\rev{\resizebox{\textwidth}{!}{\begin{tabular}{@{}lcccccc@{}}
\toprule
Regime & $\sigma_{\mathrm{vol}}/\sigma_{\mathrm{stab}}$ & $\bar\mu=\lambda_2^{\mathrm{exo}}/2$ (day$^{-1}$) & $\xi\gamma$ & $\Delta$ mean PnL (norm.\ 95\% CI) & Sharpe-like (van., eq.) & Eq.\ wins \\
\midrule
Calm              & $1.50$ & $1.9$  & $0.1$ & $+55\%\ (\pm 7)$  & $0.78,\ 1.10$ & $88\%$ \\
Base (calibrated) & $3.05$ & $3.83$ & $0.2$ & $+86\%\ (\pm 12)$ & $0.63,\ 0.99$ & $85\%$ \\
Volatile          & $4.32$ & $7.7$  & $0.2$ & $+56\%\ (\pm 13)$ & $0.58,\ 0.84$ & $83\%$ \\
Stressed          & $6.05$ & $15.3$ & $0.4$ & $+76\%\ (\pm 13)$ & $0.69,\ 0.99$ & $83\%$ \\
\bottomrule
\end{tabular}}}
\end{table}

\begin{figure}[htbp]
    \centering
    \includegraphics[width=\linewidth]{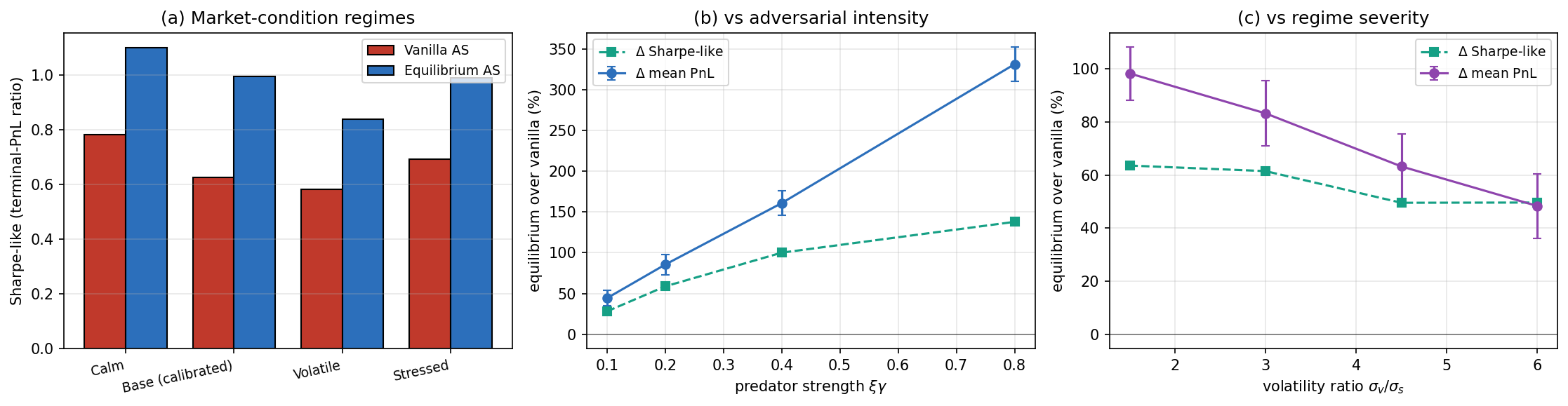}
    \caption{\rev{Process-parameter robustness sweep (calibrated once from the seven-day Kraken data; error bars are normalized paired-difference $95\%$ CIs). (a) Sharpe-like ratio of vanilla vs.\ equilibrium AS across the four market-condition regimes. (b) Equilibrium-over-vanilla improvement vs.\ predator strength $\xi\gamma$: it grows monotonically, the signature of Risk Isomorphism. (c) Improvement vs.\ volatility ratio $\sigma_{\mathrm{vol}}/\sigma_{\mathrm{stab}}$: it \emph{shrinks} as volatility dwarfs the fixed predator correction. The switching-speed axis (mean rate $\bar\mu=\lambda_2^{\mathrm{exo}}/2$, with stationary occupancy held fixed; reported in the text) varies the exogenous spectral gap directly.}}
    \label{fig:robustness}
\end{figure}

\section{Conclusion}\label{sec:conclusion}

\rev{This paper presents a hierarchical ``games-in-games'' control framework for systems governed by \emph{regime-switching jump-diffusions}. By decomposing the problem into a fast inner game and a strategic outer game, we derived a coupled system of Hamilton-Jacobi-Isaacs (HJI) equations via a unified Dynkin identity and gave a feedback Stackelberg verification principle for the bilevel solution. For the LQ specialization this yields a coupled Riccati flow, for which we established symmetry, well-posedness, and a sufficient condition for global existence; we related this flow to the Lyapunov--Metzler and Riccati--Metzler frameworks, recovering them as its non-adversarial, frozen-rate special cases, and characterized the spectral dissipation of the closed-loop generator selected by the outer game. In the Exponential-Affine market-making application the transformed inventory system admits a matrix-exponential representation.}

The framework's practical value is demonstrated through an adversarial market microstructure case study. The results reveal a \emph{Risk Isomorphism} principle, where the hierarchical controller naturally interprets strategic predation as effective volatility, inducing an anticipatory equilibrium that pre-emptively widens spreads. Future research will extend this architecture to partially observable regimes\rev{, where the inner controller must act on a filtered or asynchronously estimated mode, connecting to recent advances in asynchronous filtering and estimation for Markov jump systems~\citep{FangRenWangStojanovicHe2024},} and integrate data-driven learning for empirical transition kernels.

\section*{Declaration of generative AI and AI-assisted technologies in the writing process}
\rev{During the preparation of this work the authors used Claude Opus 4.7 (Anthropic) to assist with the numerical simulation study, and ChatGPT 5.5 (OpenAI) for minor mathematical verification and language editing. After using these tools, the authors reviewed and edited the content as needed and take full responsibility for the content of the publication.}

\end{document}